\documentclass[11pt]{article}

\usepackage[letterpaper,margin=1in]{geometry}
\usepackage[T1]{fontenc}
\usepackage[utf8]{inputenc}
\usepackage{lmodern}
\usepackage{amsmath,amssymb,amsthm,mathtools,mleftright}
\usepackage{thmtools,thm-restate}
\usepackage[normalem]{ulem}
\usepackage{xcolor}

\usepackage[ruled,vlined,linesnumbered]{algorithm2e}
\usepackage[hidelinks]{hyperref}
\usepackage[nameinlink,noabbrev]{cleveref}
\usepackage[nocompress,sort]{cite}
\usepackage{titling}

\AtBeginDocument{%
  \DeclareFontShape{T1}{lmr}{m}{scit}
    {<->ssub*lmr/m/scsl}{}%
}

\newtheorem{theorem}{Theorem}
\newtheorem{lemma}{Lemma}
\newtheorem{corollary}{Corollary}
\theoremstyle{remark}
\newtheorem{remark}{Remark}

\makeatletter
\renewcommand{\paragraph}{%
  \@startsection{paragraph}{4}%
  {\z@}{1.6ex \@plus 1ex \@minus .2ex}{-0.2em}%
  {\normalfont\normalsize\bfseries}%
}
\makeatother

\newcommand\thisenumsymbol{}

\newlength{\bibitemsep}
\newlength{\bibparskip}
\let\oldthebibliography\thebibliography
\renewcommand\thebibliography[1]{%
  \oldthebibliography{#1}%
  \setlength{\parskip}{\bibitemsep}%
  \setlength{\itemsep}{\bibparskip}%
}

\makeatletter
 \newcommand{\linkdest}[1]{\Hy@raisedlink{\hypertarget{#1}{}}}
\makeatother

\newcommand{\E}{\mathbb E}
\newcommand{\bF}{\mathbb F}
\newcommand{\Prb}{\mathbb P}
\newcommand{\OPT}{\operatorname{OPT}}
\newcommand{\ALG}{\mathrm{ALG}}
\newcommand{\spn}[1]{\left\langle #1 \right\rangle}

\newcommand{\supp}{\operatorname{supp}}

\def\dif#1{\mathop{d #1}}

\DeclareMathOperator{\clo}{cl}
\DontPrintSemicolon
\SetAlFnt{\small}
\SetAlgoCaptionSeparator{.}
\SetKwFunction{Update}{UpdateDistributions}
\hypersetup{pdftitle={The Strong Secretary Conjecture is True for Linear Matroids},pdfauthor={Kristof Berczi, Shaddin Dughmi, Vasilis Livanos, Jose Soto, Victor Verdugo}}

\def\floor#1{\left\lfloor #1 \right\rfloor}

\def\set#1{\left\{ #1 \right\}}

\newcommand{\prn}[2][\!]{\mleft({#2}\mright)}
\def\brk#1{\left[#1\right]}
\newcommand{\1}[1]{\mathbf{1}\mleft[ #1 \mright]}
\def\midd{\:\middle|\:}

\title{The Strong Secretary Conjecture is True
for Linear Matroids}

\thanksmarkseries{alph}\newcommand{\authornewrow}{%
  \end{tabular}\par\medskip
  \begin{tabular}[t]{c}%
}

\author{
    Krist\'of B\'erczi
    \thanks{
        HUN-REN--ELTE Egerv\'ary Research Group, Department of Operations Research, E\"otv\"os Lor\'and University, and HUN-REN Alfr\'ed R\'enyi Institute of Mathematics, Budapest, Hungary. {\tt kristof.berczi@ttk.elte.hu}. 
    }
    \and Shaddin Dughmi
    \thanks{
        Department of Computer Science, University of Southern California, Los Angeles, USA. {\tt shaddin@usc.edu}.
    }
    \and Vasilis Livanos
    \thanks{
        Department of Computer Science, University of Southern California, Los Angeles, USA. {\tt vas.livanos@gmail.com}.
    }
    \authornewrow 
    Jos\'e A. Soto
    \thanks{
        Department of Mathematical Engineering and Center for Mathematical Modeling, Universidad de Chile, Santiago, Chile. {\tt jsoto@dim.uchile.cl}
    }
    \and Victor Verdugo
    \thanks{
        Institute for Mathematical and Computational Engineering, and Department of Industrial and Systems Engineering, Pontificia Universidad Católica de Chile, Santiago, Chile. {\tt victor.verdugo@uc.cl} 
    }
}
\date{}

\begin{document}
\maketitle
\thispagestyle{empty}

\begin{abstract}
We prove a $1/e$ guarantee for Matroid Secretary on linear matroids: each element of a fixed optimal basis is selected with probability at least $1/e$. This settles the strong secretary conjecture for this class of matroids. Our result holds both when the matroid is known in advance and when a linear representation is given online. In the known-matroid model, the result extends to fully modular extendable matroids.

We also obtain $1/e$ guarantees for a more general model on linear matroids, both when the matroid is known and when its representation is revealed online. In XOS Matroid Secretary, agents with XOS valuations over the elements arrive in uniformly random order, and the algorithm must assign each arriving agent a bundle of elements. The bundles must be pairwise disjoint, and their union must be an independent set of the matroid. We obtain an expected total value of at least $1/e$ times the offline optimum. This model contains Group Matroid Secretary, in which the elements are adversarially partitioned into groups that arrive in uniformly random order, and the algorithm selects a subset of each arriving group. For this special case, we select each element of a fixed optimal basis with probability at least $1/e$.
\end{abstract}

\newpage
\pagenumbering{roman}
\tableofcontents
\newpage
\pagenumbering{arabic}
\setcounter{page}{1}

\section{Introduction}
\label{sec:intro}

The \emph{Matroid Secretary Problem}, introduced by Babaioff,
Immorlica, and Kleinberg~\cite{BIK2007}, extends the classical secretary
problem from selecting a single element to selecting an independent set of
a matroid. Each element $e$ has a fixed nonnegative weight $w(e)$, and the
elements arrive in uniformly random order, with weights revealed upon
arrival. The algorithm must irrevocably decide whether to accept the element
while keeping the selected set independent.
Let $E$ denote the ground set. We fix a total order on $E$ before the arrivals and use it to break weight ties in every greedy computation. For every $Y\subseteq E$, let $\OPT(Y)$ denote the resulting greedy maximum-weight basis of the matroid restricted to~$Y$. We say that an algorithm is \emph{$c$-probability-competitive} if
\begin{equation*}
\Prb\brk{e\in\ALG}\ge c\;\text{ for every }e\in\OPT(E),
\end{equation*}
where the probability is over the uniformly random arrival order and the internal randomness of the algorithm.
For nonnegative weights, this implies an expected weight of at least $c$
times $w(\OPT(E))$. In the classical secretary problem, which corresponds
to a rank-one matroid, the best probability guarantee that holds for every number of arrivals is $1/e$.

A matroid $M$ on ground set $E$ is \emph{linear} if there
is a field $\bF$ and an indexed family of vectors
$(v_e)_{e\in E}$ over $\bF$ such that $I\subseteq E$ is
independent in $M$ if and only if $(v_e)_{e\in I}$ is
linearly independent. We call this family a
\emph{representation} of $M$ over $\bF$.
We prove that the optimal $1/e$-probability-competitive guarantee is achievable in the
\emph{online-representation model}, in which a representation of the
matroid over a finite field is revealed online. The algorithm knows
the number of elements and the underlying finite field in advance, but not
the representation. Upon arrival, each element reveals its weight and its
vector, expressed in a common coordinate system over that field. Thus, the
algorithm learns the representation one vector at a time as the elements arrive. In particular, the algorithm does not know the matroid, or a representation of it, before the arrival process begins. The algorithm is ordinal: it uses the weights only through relative comparisons, with ties broken by the fixed total order. We place no restriction on the computation performed by the algorithm. Upon each arrival, it may use any function of the information revealed so far. In particular, we do not bound its running time.

\begin{theorem}\label{thm:main}
In the online-representation model over a finite field, there is an ordinal
online algorithm that always selects an independent set $\ALG$ and satisfies
\begin{equation*}
\Prb\brk{e\in\ALG}\ge 1/e
\;\text{ for every }e\in\OPT(E).
\end{equation*}
Consequently, $\E[w(\ALG)]\ge w(\OPT(E))/e.$
\end{theorem}

Theorem~\ref{thm:main}
also holds when the vectors are revealed online over an arbitrary
field (see Remark~\ref{rem:arbitrary-fields} at the end of
Section~\ref*{sec:feasibility}).
The same guarantee is achievable for every linear matroid
in the \emph{known-matroid model}. Here the entire matroid is known before
the arrival process, while only the weights are revealed online. By a
theorem of Rado~\cite{Rado1957}, every finite matroid representable over
some field also admits a representation over a finite field. Therefore, we may compute such a
representation and apply Theorem~\ref{thm:main}. In particular, no
representation need be supplied as part of the input, and the result applies
to matroids representable over an arbitrary field.

The main idea behind our algorithm is to control how much of each subspace is spanned by the accepted vectors. The difficulty is that accepting an element affects all later elements at once: after $e$ is accepted, every element of $\spn{\ALG}+\spn{e}$ becomes unavailable, where $\spn{\ALG}$ denotes the linear span of the vectors selected by the algorithm. It is therefore not enough to control, for each element, the probability that it can still be accepted. Instead, for every subspace $L$, we maintain an upper bound on the expected dimension of $L\cap\spn{\ALG}$, which we call the \emph{dimension invariant}.

We then choose acceptance probabilities that preserve the dimension invariant while ensuring that an arriving element is accepted with a prescribed probability whenever it belongs to $\OPT(X)$, where $X$ is the set of elements revealed so far. The probability is over both the preceding arrival order and the internal randomness of the algorithm. The above requirements placed on acceptance probabilities lead to a large linear program (LP), and the main technical step in our proof is to establish its feasibility. Duality and uncrossing reduce the feasibility argument to chains of subspaces, for each of which we give an explicit feasible solution. The uncrossing step uses that, for every fixed subspace $Z$, the dimension of $Z\cap L$ is a supermodular function of the subspace~$L$. This follows from the fact that dimension is modular on subspaces: $\dim(A)+\dim(B)=\dim(A\cap B)+\dim(A+B)$. The rank function of a general matroid is only submodular, meaning this identity holds only as an inequality, which is why our argument relies on a linear representation. More generally, the same argument applies to known matroids that are restrictions of finitary modular matroids, whose lattices of flats satisfy the modular rank identity. In particular, it applies to the fully modular extendable matroids studied by B\'erczi, Geh\'er, Imolay, Lov\'asz, Padr\'o and Schwarcz~\cite{BercziEtAl}. This class includes all matroids of rank at most~$3$ and all skew-representable matroids, that is, matroids representable over a division ring.

The same approach extends to a more general arrival model on linear matroids. In XOS Matroid Secretary, agents with XOS valuations arrive in uniformly random order and receive disjoint bundles whose union is independent. We obtain an expected total value of at least $1/e$ times the offline optimum (\Cref{thm:xos} in \Cref{sec:xos}). This model contains Group Matroid Secretary, in which the elements are partitioned into groups that arrive in uniformly random order. For this special case, we obtain the stronger guarantee that each element of a fixed greedy optimum is selected with probability at least $1/e$.
The extension to agent arrivals follows a similar blueprint to our main result.
We select a subset of the bundle assigned to the arriving agent
in an optimal allocation of the observed agents. The dimension
invariant, together with a polymatroid argument, allows us
to choose a feasible random subset with prescribed element-wise
marginals. We use an LP to impose these marginals and preserve
the invariant. Duality and uncrossing again reduce its feasibility
to chains of subspaces. For each chain, we add auxiliary vectors
to the bundle before selecting the random subset, then discard
those vectors.
For XOS valuations, these marginals give a corresponding fraction
of the bundle's value in expectation. For groups, each element
of the global greedy optimum also belongs to the greedy optimum
of the observed elements when its group arrives. We can therefore
use the same marginals to bound its selection probability. 

We first present our main result for the Matroid Secretary Problem in \Cref{sec:algorithm,sec:feasibility}.
In \Cref{sec:extensions}, we show how the same guarantee extends to fully modular extendable matroids. We then develop the XOS and group arrival case independently in \Cref{sec:xos}, replicating the shared steps so that the proof is self-contained.

\subsection{Related Work}
\label{sec:related-work}

The classical secretary problem was studied by Lindley~\cite{lindley}
and Dynkin~\cite{dynkin}. Kleinberg~\cite{klein-secr} studied the
multiple-choice secretary problem, in which up to $k$ elements may be
selected. Later, Babaioff, Immorlica, and Kleinberg~\cite{BIK2007} introduced the
Matroid Secretary Problem; see also the journal version with
Kempe~\cite{mat-sec}. They conjectured that a constant expected-weight
(or utility-competitive) guarantee is achievable for every matroid.

For general rank-$r$ matroids, Babaioff, Immorlica, and Kleinberg~\cite{BIK2007} and Babaioff, Immorlica, Kempe, and Kleinberg~\cite{mat-sec} obtained an $\Omega(1/\log r)$ expected-weight guarantee. This was improved
to $\Omega(1/\sqrt{\log r})$ by Chakraborty and
Lachish~\cite{sqrt-log-r-matroid-sec}, and then to
$\Omega(1/\log\log r)$ by Lachish~\cite{loglogrank-matroid-sec1}.
Feldman, Svensson, and Zenklusen~\cite{loglogrank-matroid-sec2} later gave
a simpler algorithm achieving the same guarantee.

A stronger notion is probability-competitiveness, which requires a guarantee
for each element of the optimum. Bateni, Hajiaghayi, and
Zadimoghaddam~\cite{bateni-secretary} obtained an
$\Omega(1/\log^2 r)$ probability-competitive guarantee, and Soto,
Turkieltaub, and Verdugo~\cite{forbidden-paper} improved it to
$\Omega(1/\log r)$. The strong Matroid Secretary Conjecture asks whether
the optimal rank-one guarantee $1/e$ can be achieved under this criterion
for every matroid~\cite{forbidden-paper}. For linear matroids, we obtain the $1/e$
probability-competitive guarantee in the known-matroid and
online-representation models.
Constant guarantees were previously known for
transversal~\cite{trans-secretary,KesselheimRTV13,forbidden-paper},
graphic~\cite{babaioff-secretary-2,BIK2007,banihashem2025beating,
labeling-schemes,korula-pal-graphic,forbidden-paper},
laminar~\cite{labeling-schemes,zahra-laminar-secretary,ImW11,
free-order-secretary,ma-tang-secretary,forbidden-paper}, and
regular matroids~\cite{dinitz-secretary}.

\paragraph{Recent and concurrent results. }
Several results on the Matroid Secretary Problem appeared in rapid succession over a few weeks.
Singla~\cite{Singla2026} resolved the original conjecture with an ordinal $1/4$-probability-competitive algorithm for arbitrary matroids in the oracle model, requiring only $n$ in advance. His approach reduces the problem to a matroid version of the random-order two-sided Game of Googol~\cite{correa-googol} via reversible updates that preserve greedy decisions on finalized elements; using the reduction framework of Feldman, Svensson, and Zenklusen~\cite{FeldmanSZ22}, he also obtained an $\Omega(1/k^3)$-competitive algorithm for the intersection of $k$ matroids.
Building on Singla's approach, Chan, Lin, and Wang~\cite{ChanLW2026} improved the selection probability to $1/c \approx 1/3.1462$, where $c>1$ satisfies $c-\ln c=2$, via random deletions and time-dependent acceptance.

Independently and concurrently, Abdi, Banihashem, Hajiaghayi, and Mittal~\cite{abdi2026strong} obtained an ordinal $1/e$-probability-competitive algorithm for known linear matroids using linear programming, with a feasibility proof based on the same subspace dimension invariant and uncrossing argument as ours. They also proved a tight $1/2$ single-sample prophet inequality for arbitrary matroids, yielding a $1/64$-competitive secretary algorithm in the oracle model.

Finally, Huang~\cite{Huang2026} settled the strong Matroid Secretary Conjecture for arbitrary matroids with an ordinal, expected-polynomial-time $1/e$-probability-competitive algorithm in the oracle model, matching the optimal rank-one bound via a continuous-time formulation.

The algorithms of Singla and Huang apply to arbitrary matroids using only independence-oracle access, but they process one element at a time. They do not directly apply when elements arrive in groups or are allocated to agents. In Group Matroid Secretary, the groups arrive in random order, but the elements do not: the elements of a group arrive together. In XOS Matroid Secretary, an arriving agent receives a bundle, so several elements must be selected jointly. Our dimension bounds hold for every subspace, not only for the lines spanned by single elements. Together with a polymatroid argument, they allow us to select a random subset of a bundle with prescribed marginals while preserving the invariant. We thus obtain $1/e$ guarantees for both models on linear matroids. These guarantees are optimal, since both models contain the classical secretary problem. To our knowledge, no constant-competitive algorithms were previously known for these models on general linear matroids.

\paragraph{Linear-programming approaches. }
\,Linear programming has been previously used to design and analyze secretary
algorithms. Buchbinder, Jain, and Singh~\cite{BuchbinderJS14} gave a finite
LP characterization of secretary algorithms, and Chan, Chen, and
Jiang~\cite{ChanCJ15} developed a continuous LP approach for generalized
secretary problems. Related LP formulations have since been used in
secretary problems with advice and in sample-driven optimal
stopping~\cite{DuttingLLV24,CorreaCES24}.

\paragraph{Matroid information model. }
\,The information available about the matroid gives rise to two standard
variants~\cite{GharanV13,matroid-embeddings}. In the
\emph{known-matroid model}, the ground set and the entire independence
structure are available before the arrival process, and only the weights are
revealed online. In the \emph{unknown-matroid model}, the algorithm knows the
number of elements in advance but can query independence only for subsets of
elements that have already arrived. In particular, it has no information
about dependencies involving unseen elements. In the online-representation model considered in this paper,
the representation is not known in advance, but each arriving
element reveals its vector in a common coordinate system.
From these vectors, we can answer every
independence query on observed elements. We can also construct
subspaces of the observed span and choose auxiliary vectors in them.
The subspaces need not be spanned by observed elements, and the
auxiliary vectors need not correspond to elements of the matroid.

For many special classes of matroids, the natural online model similarly
reveals a representation of each arriving element. For graphic matroids, an
arriving edge reveals its endpoints, while for transversal matroids an
arriving element reveals its neighborhood in the underlying bipartite
representation. This additional information has been important in obtaining
constant guarantees. For example, the $1/e$ guarantee of Kesselheim,
Radke, T{\"o}nnis, and V{\"o}cking~\cite{KesselheimRTV13} for transversal
matroids uses this bipartite representation. Similarly, the strongest
previous guarantees for graphic matroids use the graph
representation~\cite{banihashem2025beating,labeling-schemes}, whereas
D{\"u}tting, Paes Leme, P{\'a}l, and Patel~\cite{DuttingEtAl2026} recently
obtained a $1/36$ utility-competitive guarantee using only independence
queries on arrived elements. Our online-representation model takes the
analogous viewpoint for linear matroids.

Cristi, D{\"u}tting, Kleinberg, Paes Leme, and
Patel~\cite{matroid-embeddings} introduced online matroid embeddings, which
map a matroid revealed online into a larger known matroid. They constructed
embeddings for binary and laminar matroids and obtained reductions from
unknown to known matroids with an arbitrarily small loss in the
expected-weight guarantee. They also proved limitations of this embedding
approach for general matroids.

\paragraph{Other approaches and models. }
\,Dughmi~\cite{dughmi1,dughmi2} used duality to establish an equivalence
between constant-competitive matroid secretary algorithms and random-order
contention resolution for possibly correlated inputs. On the negative side,
Bahrani, Beyhaghi, Singla, and Weinberg~\cite{BahraniBSW21} proved
limitations of broad classes of greedy and partition-based algorithms.
Abdolazimi, Karlin, Klein, and Oveis Gharan~\cite{matroid-partition} showed
that the partition-property approach cannot give a constant guarantee for
general binary matroids. These barriers concern the specified algorithmic
approaches rather than the existence of constant-competitive algorithms.

Constant guarantees were also known under other online models.
Soto~\cite{soto-secretary} studied random assignment of weights to elements,
and Oveis Gharan and Vondr{\'a}k~\cite{GharanV13} extended this line to
adversarial arrival orders. Santiago, Sergeev, and
Zenklusen~\cite{SantiagoSZ23} obtained a constant guarantee for random
assignment and random arrival order without knowing the matroid. Jaillet,
Soto, and Zenklusen~\cite{free-order-secretary} considered the free-order
model, where the algorithm chooses the order in which elements are
inspected.

\paragraph{Combinatorial allocations and group arrivals. }
\,Several works study secretary problems beyond single-element arrivals with additive weights. When individual elements arrive online and the goal is to maximize a global submodular function over the selected independent set, Feldman and Zenklusen~\cite{FeldmanZenklusen2018} reduced the problem to the linear-objective case (see also~\cite{bateni-secretary,ma-tang-secretary}). Combined with recent algorithms for general matroids, this yields constant guarantees for nonnegative submodular objectives~\cite{Singla2026}; see also Rubinstein and Singla~\cite{RubinsteinS17} for combinatorial prophet and secretary models optimizing XOS and subadditive functions over a single matroid. For settings where arriving agents receive bundles of elements, Kesselheim, Radke, T{\"o}nnis, and V{\"o}cking~\cite{KesselheimRTV13} obtained an optimal $1/e$-competitive algorithm for online combinatorial auctions with submodular and XOS~\cite{LehmannLN06,Feige09} valuations, corresponding to disjoint bundle allocations without matroid constraints. Related random-order and sample-based allocation models for packing LPs and XOS valuations were studied in~\cite{KesselheimRTV18,EhsaniHKS18,DuttingKLRS24}. For unit-demand agents, the setting of~\cite{KesselheimRTV13} specializes to one-sided online weighted bipartite matching, which Marinkovic, Soto, and Verdugo~\cite{marinkovic-soto-verdugo} extended to general independence systems on the items, including hypergraph matchings and matroid intersections admitting directed certifiers. In a different direction, Korula and P{\'a}l~\cite{korula-pal-graphic} introduced secretary problems with group arrivals, in which an adversary partitions the elements into groups that arrive in uniformly random order and the algorithm may select multiple elements from each arriving group, and studied this model for graph and hypergraph matching.

\section{The Dimension-Invariant Algorithm}
\label{sec:algorithm}

In this section, we describe our algorithm for the linear matroid secretary problem in the online-representation model. Let $n$ be the number of elements, and fix a sample size $1\leq k<n$. The first $k$ arrivals form the sample and are rejected. In the classical rank-$1$ secretary algorithm, if the best element arrives in position $i>k$, then it is selected precisely when the best of the first $i-1$ elements appears in the sample, which happens with probability $k/(i-1)$. We seek the same conditional probability in a matroid: if $X$ is the set of the first $i$ observed elements and $e\in\OPT(X)$ arrives last, we want to accept $e$ with probability $k/(i-1)$. As in the rank-$1$ case, choosing $k$ appropriately and averaging over the arrival position yields a $1/e$ guarantee. We obtain these acceptance probabilities by solving linear programs that simultaneously enforce the desired acceptance probabilities and control, for every subspace $L$, the expected dimension of $L\cap\spn{\ALG}$.

\subsection{The Dimension Invariant}
\label{sec:invariant}

When no confusion may arise, we identify an observed element with its
representing vector. We use the notation $\spn{\cdot}$ for \emph{linear
span}. If $X$ is a set of observed elements, then $\spn{X}$ denotes the
span of their representing vectors, and $\spn{e}$ denotes the subspace
spanned by the vector representing $e$. We write $W\leq V$ when $W$ is
a subspace of $V$, and for subspaces $W_1,W_2$ we write
$W_1+W_2\coloneqq\spn{W_1\cup W_2}$.
For a set $X$ and an element $e$, we use the shorthand
$X-e\coloneqq X\setminus\{e\}$ and $X+e\coloneqq X\cup\{e\}$. When an
element $e$ arrives, let $X$ be the set of elements observed up to and
including $e$. We call $e$ \emph{improving} if $e\in\OPT(X)$.

After the sample, the \emph{state} of the algorithm consists of an observed set $X$, an accepted set $\ALG$, and the span $\spn{\ALG}$ of the accepted vectors. In addition to its state, our algorithm maintains some auxiliary information for every set $Y\subseteq X$. We store a distribution $\mu_Y$ supported on the subspaces of $\spn{Y}$, for which we show the following: conditioned on $Y$ being the first $|Y|$ arrivals, $\mu_Y(W)$ is the probability that $\spn{\ALG}=W$ after the elements of $Y$ have been processed. We also store a conditional acceptance probability $p_Y(e,W)$ for every $e\in Y$ and $W\in\supp(\mu_{Y-e})$. We use $p_Y(e,W)$ when the elements of $Y$ arrive first, $e$ is the last among them, and the accepted elements span $W$ immediately before its arrival.
Recursive dependencies necessitate that we maintain this information for all sets $Y \subseteq X$, not merely those occurring as prefixes of the realized arrival order: $\mu_Y$ is evidently a function of $p_Y$ and $\mu_{Y-e}$ for all $e \in Y$, and moreover we compute $p_Y$ as a function of all the $\mu_{Y-e}$ in order to satisfy certain properties.

For every $Y\subseteq X$ with $|Y|\leq k$, if the elements of $Y$ were the first $|Y|$ arrivals, the algorithm would reject all of them as part of the sample. Therefore, we set $\mu_Y(\{0\})=1$, $\mu_Y(W)=0$ for every $W\neq\{0\}$, and $p_Y(e,W)=0$ for every $e\in Y$ and $W\in\supp(\mu_{Y-e})$. Now suppose that $|X|=i>k$. 
For every $e\in\OPT(X)$, we want to accept $e$ with probability
$k/(i-1)$, conditional on $X$ being the set of the first
$i$ arrivals and $e$ arriving last among $X$. 
 We maintain a particular invariant which ensures this remains possible, which we call the \emph{Dimension Invariant Inequality.} 
In particular, for every set $X$  with $|X|\geq k$ whose distribution $\mu_X$ has been computed, we require
that
\begin{equation}\label{eq:dimension-invariant}
    \E_{W \sim \mu_X}\brk{\dim(W \cap L)}
    \leq \prn{1 - \frac{k}{|X|}} \dim(L)
    \quad \text{for every subspace $L$.}
\end{equation}
To see why we use the factor $1-k/|X|$, let $|Y|=i>k$ and suppose that
the last element $e$ to arrive from $Y$ is improving. Apply
\eqref{eq:dimension-invariant} to $Y-e$ and $L=\spn{e}$. Since $e$ is
not a loop, $\spn{e}$ has dimension one. Moreover, for a random
subspace $W$ distributed according to $\mu_{Y-e}$,
$\dim(W\cap\spn{e})$ is $1$ if $e\in W$ and $0$ otherwise. Hence
\eqref{eq:dimension-invariant} implies that
$\Prb\brk{e\notin W}\geq k/(i-1)$. Thus, before $e$ arrives, it can be added to $\ALG$ without violating
independence with at least the desired probability. We now need to
choose the acceptance probabilities so that every improving element is
accepted with exactly this probability and the Dimension Invariant
Inequality continues to hold.
Notice that it suffices to guarantee the Dimension Invariant Inequality
for $L\leq\spn{X}$. This is because every $W\in\supp(\mu_X)$ satisfies
$W\leq\spn{X}$, and thus $W\cap L=W\cap(L\cap\spn{X})$. Therefore,
applying \eqref{eq:dimension-invariant} to $L\cap\spn{X}$ gives the
inequality for an arbitrary ambient subspace $L$.

Fix the current observed set $X$, and let $Y\subseteq X$
with $|Y|=i>k$. Suppose that $\mu_A$ has been computed
for every $A\subsetneq Y$.
Suppose also that acceptance probabilities $p_Y(e,W)\in[0,1]$
have been chosen for every $e\in Y$ and
$W\in\supp(\mu_{Y-e})$, and write $p$ for their collection.
We define $\mu_Y$ as the distribution of the span obtained
by choosing $e\in Y$ uniformly, sampling $W\sim\mu_{Y-e}$,
and accepting $e$ with probability $p_Y(e,W)$.
If we reject $e$, the span remains $W$.
If we accept $e$, the span becomes $W+\spn{e}$.
Thus, for every $U\leq\spn{Y}$,
\begin{equation}\label{eq:solution-distribution}
    \mu_Y(U)
    =
    \frac{1}{i}
    \sum_{e\in Y}
    \sum_{W\in\supp(\mu_{Y-e})}
    \mu_{Y-e}(W)
    \brk{(1-p_Y(e,W))\1{U=W} + p_Y(e,W)\1{U=W+\spn{e}}}.
\end{equation}
Since each $\mu_{Y-e}$ is a probability distribution and
$p_Y(e,W)\in[0,1]$, the values in
\eqref{eq:solution-distribution} are nonnegative and sum to one over
all subspaces $U\leq\spn{Y}$. Thus, $\mu_Y$ is a probability
distribution on the subspaces of $\spn{Y}$.

For every $e\in Y$ and every subspace $L$, let
\[
u_e(L)
=
\E_{W\sim\mu_{Y-e}}\brk{\dim(W\cap L)}
\]
and
\[
\Delta_e(W,L)
=
\dim((W+\spn{e})\cap L)-\dim(W\cap L).
\]
Thus $u_e(L)$ is the expected dimension of $W\cap L$ under
$\mu_{Y-e}$, and $\Delta_e(W,L)$ is the change in this dimension when
$W$ is replaced by $W+\spn{e}$.
To impose \eqref{eq:dimension-invariant} on $\mu_Y$, define
\begin{equation}\label{eq:violation}
g_p(L)
\coloneqq
-(i-k)\dim(L)
+\sum_{e\in Y}u_e(L)
+\sum_{e\in Y}
\E_{W\sim\mu_{Y-e}}
\brk{p_Y(e,W)\Delta_e(W,L)}.
\end{equation}
For each $L$, the function $g_p(L)$ is affine in $p$.
The following lemma shows that $\mu_Y$ satisfies
\eqref{eq:dimension-invariant} for a subspace $L$
if and only if $g_p(L)\leq0$.
\begin{lemma}\label{lem:dimension-update}
Let $Y$ be a set with $|Y|=i>k$. Suppose that, for every $e\in Y$,
$\mu_{Y-e}$ is a probability distribution supported on the subspaces
of $\spn{Y-e}$. Let
\[
    p=(p_Y(e,W):e\in Y,\ W\in\supp(\mu_{Y-e}))
\]
satisfy $p_Y(e,W)\in[0,1]$ for every $e\in Y$ and
$W\in\supp(\mu_{Y-e})$. Define $\mu_Y$ by
\eqref{eq:solution-distribution} and $g_p$ by
\eqref{eq:violation}. Then, for every subspace $L$,
\[
    \E_{U \sim\mu_Y}\brk{\dim(U\cap L)} = \frac{1}{i} \prn{(i - k) \dim(L) + g_p(L)}.
\]
Consequently, for every subspace $L$, the Dimension Invariant
Inequality for $L$ holds if and only if $g_p(L)\leq0$. In particular,
if
\[
    g_p(L)\leq0
    \qquad \forall\,L\leq\spn{Y},
\]
then $\mu_Y$ satisfies the Dimension Invariant Inequality.
\end{lemma}

\begin{proof}
By \eqref{eq:solution-distribution} and the definitions of $u_e(L)$
and $\Delta_e(W,L)$,
\begin{align*}
\E_{U\sim\mu_Y}\brk{\dim(U\cap L)}
&=
\frac{1}{i}
\sum_{e\in Y}
\left(
u_e(L)
+
\E_{W\sim\mu_{Y-e}}
\brk{p_Y(e,W)\Delta_e(W,L)}
\right) \\
&=
\frac{1}{i}
\left(
(i-k)\dim(L)+g_p(L)
\right),
\end{align*}
where the second equality follows from \eqref{eq:violation}. Since
$|Y|=i$, the Dimension Invariant Inequality for a fixed subspace $L$
is
\[
\E_{U\sim\mu_Y}\brk{\dim(U\cap L)}
\leq
\frac{i-k}{i}\dim(L).
\]
By the identity above, this holds if and only if $g_p(L)\leq0$.

Therefore, if $g_p(L)\leq0$ for every $L\leq\spn{Y}$, the Dimension
Invariant Inequality holds for every $L\leq\spn{Y}$. For an arbitrary subspace $L$, every
$U\in\supp(\mu_Y)$ satisfies
$U\cap L=U\cap(L\cap\spn{Y})$. Applying the inequality to
$L\cap\spn{Y}$ and using
$\dim(L\cap\spn{Y})\leq\dim(L)$ gives the result for $L$.
\end{proof}
We will choose the acceptance probabilities so that
$p_Y(e,W)=0$ whenever $e\notin\OPT(Y)$, that is, whenever $e$ would
not be improving if it arrived last among the elements of $Y$.
Under this restriction, \eqref{eq:violation} becomes
\[
g_p(L)
=
-(i-k)\dim(L)
+
\sum_{e\in Y}u_e(L)
+
\sum_{e\in\OPT(Y)}
\E_{W\sim\mu_{Y-e}}
\brk{p_Y(e,W)\Delta_e(W,L)}.
\]

\subsection{The Acceptance LP and the Algorithm}
\label{subsec:lp}

Fix the current observed set $X$, and let $Y\subseteq X$ with
$|Y|=i>k$. Suppose that $\mu_A$ has already been computed for every
$A\subsetneq Y$. We choose the acceptance probabilities $p_Y(e,W)$,
for $e\in Y$ and $W\in\supp(\mu_{Y-e})$, by solving the following
feasibility LP, whose variables are precisely these quantities
$p_Y(e,W)$. We write $p$ for their collection.
\begin{equation*}
\begin{alignedat}{2}
\mathrm{LP}(Y):\quad
& 0 \leq p_Y(e,W) \leq 1
&& \qquad \forall\, e\in Y,\ W\in\supp(\mu_{Y-e}), \\[2pt]
& p_Y(e,W) = 0
&& \qquad \forall\, e\in Y,\ W\in\supp(\mu_{Y-e})
   \text{ with } e\notin\OPT(Y)\text{ or }e\in W, \\[2pt]
& \E_{W\sim\mu_{Y-e}}\brk{p_Y(e,W)} = \frac{k}{i-1}
&& \qquad \forall\, e\in\OPT(Y), \\[2pt]
& g_p(L) \leq 0
&& \qquad \forall\, L\leq\spn{Y}.
\end{alignedat}
\end{equation*}
The first constraint ensures that the variables are valid probabilities. The second restricts acceptance to elements of $\OPT(Y)$ and assigns zero acceptance probability whenever $e\in W$, so that accepting $e$ preserves independence. For each $e\in\OPT(Y)$, the third constraint sets the expected value of $p_Y(e,W)$ over $W\sim\mu_{Y-e}$ equal to the target value $k/(i-1)$. 
Finally, by Lemma~\ref{lem:dimension-update}, the last constraint is equivalent to requiring that the distribution $\mu_Y$ defined by \eqref{eq:solution-distribution} satisfies the Dimension Invariant Inequality. Since the representation is over a finite field, $\spn{Y}$ has only
finitely many subspaces, so $\mathrm{LP}(Y)$ is a finite linear program (for arbitrary fields, see Remark~\ref{rem:arbitrary-fields} at the end of Section~\ref*{sec:feasibility}).

The following lemma shows that $\mathrm{LP}(Y)$ is feasible whenever
the distributions $\mu_{Y-e}$ satisfy the Dimension Invariant
Inequality \eqref{eq:dimension-invariant}. In other words, the inequality can indeed be maintained as an invariant by using acceptance probabilities feasible for this LP. Its proof is deferred to
Section~\ref{sec:feasibility}.

\begin{lemma}\label{lem:lp-feasible}
Let $Y$ be a set with $|Y|=i>k$. Suppose that, for every $e\in Y$,
the distribution $\mu_{Y-e}$ satisfies the Dimension Invariant
Inequality \eqref{eq:dimension-invariant}. Then $\mathrm{LP}(Y)$ is
feasible.
\end{lemma}

We fix a priori a rule for choosing a feasible solution of
$\mathrm{LP}(Y)$, requiring the chosen solution to be rational whenever
a rational feasible solution exists, and denote the chosen solution by
$p_Y$. We then compute $\mu_Y$ from $p_Y$ and the distributions
$\mu_{Y-e}$ using \eqref{eq:solution-distribution}. This completes the
recursive construction of $(\mu_Y,p_Y)$.
This fixed choice ensures that $(\mu_Y,p_Y)$ depends only on $Y$ and
the data revealed for its elements, and not on the order in which the
elements of $Y$ arrived. We will show in
Lemma~\ref{lem:algorithm-properties} that, throughout the execution,
the relevant linear programs have rational coefficients and admit
rational feasible solutions. Consequently, the rule above always
selects rational acceptance probabilities.

Throughout its execution, our algorithm maintains a table indexed by
subsets $Y$ of the observed elements. For each such $Y$, the table
stores the pair $(\mu_Y,p_Y)$, where
\[
\mu_Y=\bigl(\mu_Y(U):U\leq\spn{Y}\bigr)
\]
is a probability distribution on the subspaces of $\spn{Y}$, and
\[
p_Y=\bigl(p_Y(f,W):f\in Y,\ W\in\supp(\mu_{Y-f})\bigr)
\]
is the collection of acceptance probabilities associated with $Y$.
Once a pair $(\mu_Y,p_Y)$ is computed, it is stored and reused in all
subsequent calls to \textsc{UpdateDistributions}. Thus, when a new
element $e$ arrives, the subroutine only needs to compute the pairs
corresponding to sets $Y$ that contain $e$.

We can now describe our online algorithm, detailed in Algorithm~\ref{alg:online}. After each arrival, the
subroutine \textsc{UpdateDistributions} (Algorithm~\ref{alg:update}) computes the new entries of
this table. The online algorithm then uses
$p_X(e,\spn{\ALG})$ to decide whether to accept the arriving element
$e$. Since these probabilities are rational, the corresponding
acceptance decisions can be implemented exactly.

\begin{algorithm}[ht]
\caption{Linear Matroid Secretary}
\label{alg:online}
\KwIn{Number of elements $n$ and sample size $k$}
$X \gets \emptyset$; $\ALG \gets \emptyset$\;
\ForEach{arriving element $e$}{
  Observe $e$'s vector and its relative weight rank among the observed
  elements\;
  $X \gets X+e$\;
  \Update{$X,e$}\;
  \If{$|X|>k$}{
    Accept $e$ with probability $p_X(e,\spn{\ALG})$\;
    \If{$e$ is accepted}{
      $\ALG \gets \ALG+e$\;
    }
  }
}
\KwRet{$\ALG$}
\end{algorithm}

\begin{algorithm}[ht]
\caption{\textsc{UpdateDistributions}$(X,e)$}
\label{alg:update}
\KwIn{The current observed set $X$, the newly arrived element $e$,
together with the revealed vectors and relative weight order of the
elements of $X$}
\ForEach{$Y\subseteq X$ with $e\in Y$, in increasing order of $|Y|$}{
  \eIf{$|Y|\leq k$}{
    Set $\mu_Y$ to the point mass at $\{0\}$ and $p_Y=0$\;
  }{
    Compute $\OPT(Y)$ and enumerate the subspaces of $\spn{Y}$\;
    Solve $\mathrm{LP}(Y)$ using $(\mu_{Y-f})_{f\in Y}$, and let
    $p_Y$ be the chosen solution\;
    Compute $\mu_Y$ from \eqref{eq:solution-distribution}\;
  }
  Store $(\mu_Y,p_Y)$\;
}
\end{algorithm}

\begin{lemma}\label{lem:algorithm-properties}
Consider Algorithm~\ref{alg:online} together with the subroutine
\textsc{UpdateDistributions}. For every fixed set $X$, the algorithm is
well-defined up to the processing of $X$, and the following properties
hold:
\begin{enumerate}
    \renewcommand{\theenumi}{(\roman{enumi})}
    \renewcommand{\labelenumi}{\theenumi}

    \item\label{item:span-distribution}
    Conditional on $X$ being the set of the first $|X|$ arrivals, $\mu_X$ is the distribution of $\spn{\ALG}$ after the elements of
    $X$ have been processed.

    \item\label{item:independence}
    $\ALG$ is always independent.

    \item\label{item:dii}
    If $|X|\geq k$, then $\mu_X$ satisfies the Dimension Invariant
    Inequality \eqref{eq:dimension-invariant}.

    \item\label{item:acceptance-probability}
    If $|X|=i>k$ and $e\in\OPT(X)$, then, conditional on $X$ being the set of the first $i$ arrivals and $e$ arriving last among them, the algorithm accepts $e$ with
    probability $k/(i-1)$.
    \item\label{item:rationality}
All entries of $\mu_X$ and $p_X$ are rational.
\end{enumerate}
\end{lemma}

\begin{proof}
We prove simultaneously, by induction on $|X|$, that the algorithm is well defined and satisfies all five properties.

As our base case, we take $|X|\leq k$. All elements of $X$ belong to the sample and are rejected.
Hence $\spn{\ALG}=\{0\}$ and $\mu_X$ is the point mass at $\{0\}$. Moreover, $p_X=0$, so all entries of $\mu_X$ and $p_X$ are rational,
proving \ref{item:rationality}.
Thus, $\mu_X$ is the distribution of $\spn{\ALG}$, proving
\ref{item:span-distribution}, and $\ALG$ is independent, proving
\ref{item:independence}. When $|X|=k$,
\eqref{eq:dimension-invariant} holds with equality, proving
\ref{item:dii}.

Now we proceed by induction for $X$ with $|X|=i>k$. First we show the algorithm is well defined up to the processing of $X$, and moreover satisfies property \ref{item:rationality}. Let $a$ be the last element of $X$ to arrive. When
\textsc{UpdateDistributions}$(X,a)$ processes $X$, every
$\mu_{X-f}$ with $f\in X$ has already been computed: if $f=a$, it was
computed before the arrival of $a$, while if $f\neq a$, the set $X-f$
has size $i-1$ and, since sets are processed in increasing order of
cardinality, it is processed earlier in the same call to
\textsc{UpdateDistributions}. For $f\neq a$, the set $X-f$ need not be the set of the first
$i-1$ arrivals. However, by the fixed choice rule,
$(\mu_{X-f},p_{X-f})$ depends only on $X-f$ and the data revealed for
its elements, so the induction hypothesis still applies. Hence, by
\ref{item:dii} and \ref{item:rationality}, each $\mu_{X-f}$ satisfies
the Dimension Invariant Inequality and has rational entries. Lemma~\ref{lem:lp-feasible} therefore implies that
$\mathrm{LP}(X)$ is feasible. Moreover, $\mathrm{LP}(X)$ has rational
coefficients, so it admits a rational feasible solution. By our fixed
choice rule, $p_X$ has rational entries, and
\eqref{eq:solution-distribution} then implies that $\mu_X$ has rational
entries. Thus, $p_X$ and $\mu_X$ are well-defined, and
\ref{item:rationality} holds.

For property \ref{item:span-distribution}, fix $e\in X$ and condition on $e$ being the last element to arrive
among the elements of $X$. The relative order of the elements in
$X-e$ is then uniformly random. Thus, by induction, immediately before
the arrival of $e$, $\spn{\ALG}$ has distribution $\mu_{X-e}$.
Now fix $W\in\supp(\mu_{X-e})$. Conditional on
$\spn{\ALG}=W$ immediately before the arrival of $e$, the algorithm
rejects $e$ with probability $1-p_X(e,W)$, in which case
$\spn{\ALG}$ remains equal to $W$, and accepts $e$ with probability
$p_X(e,W)$, in which case $\spn{\ALG}$ becomes $W+\spn{e}$.
Therefore, conditional on $e$ being the last element to arrive, the
probability that $\spn{\ALG}=U$ after processing $e$ is
\[
\sum_{W\in\supp(\mu_{X-e})}
\mu_{X-e}(W)
\brk{
(1-p_X(e,W))\1{U=W\}}
+
p_X(e,W)\1{U=W+\spn{e}}
}.
\]
Every element of $X$ is equally likely to arrive last. Averaging the
expression above over $e\in X$ gives exactly
\eqref{eq:solution-distribution} with $Y=X$. Hence $\mu_X$ is the
distribution of $\spn{\ALG}$ after processing the elements of $X$.
This proves \ref{item:span-distribution}.

We next prove that $\ALG$ remains independent (property \ref{item:independence}). Let $e$ be the last
element of $X$ to arrive. By induction, $\ALG$ is
independent before $e$ arrives. If $e$ is accepted,
the second constraint of $\mathrm{LP}(X)$ implies that
$e\notin\spn{\ALG}$. Therefore, $\ALG+e$ is independent. This proves
\ref{item:independence}.

For the Dimension Invariant Inequality (property \ref{item:dii}), the last constraint of
$\mathrm{LP}(X)$ gives
\[
    g_{p_X}(L)\leq 0
    \qquad \forall\,L\leq\spn{X}.
\]
By Lemma~\ref{lem:dimension-update}, $\mu_X$ satisfies
\eqref{eq:dimension-invariant} for every $L\leq\spn{X}$. As observed in
Section~\ref{sec:invariant}, this implies the inequality for every
subspace $L$. This proves \ref{item:dii}.

Finally, we prove that elements in $\OPT(X)$ are accepted with the desired probability $k/(i-1)$ conditioned on arriving last (Property \ref{item:acceptance-probability}). Fix $e\in\OPT(X)$ and condition on $e$ being the last element
to arrive among the elements of $X$. The relative order of the
elements in $X-e$ is uniformly random. By induction, immediately
before the arrival of $e$, $\spn{\ALG}$ has distribution
$\mu_{X-e}$. When $\spn{\ALG}=W$, the algorithm accepts $e$ with
probability $p_X(e,W)$, using fresh randomness. Therefore,
\begin{equation}\label{eq:conditional}
    \Prb\brk{e\in\ALG \midd e\text{ arrives last among }X} = \E_{W\sim\mu_{X-e}}\brk{p_X(e,W)} = \frac{k}{i-1},
\end{equation}
where the last equality follows from the third constraint of
$\mathrm{LP}(X)$. This proves \ref{item:acceptance-probability}.
\end{proof}

\subsection{The \texorpdfstring{$1/e$}{1/e} Guarantee}
\label{sec:guarantee}

If $n=1$ or $n=2$, we do not run Algorithm~\ref{alg:online}. Instead,
we accept the first non-loop element that arrives. Every element of
$\OPT(E)$ is then selected with probability at least $1/n\geq 1/e$.

Hence assume that $n\geq3$. We run Algorithm~\ref{alg:online} with
sample size $k=\floor{n/e}$. Fix an element $e\in\OPT(E)$. For every
$X\subseteq E$ containing $e$, we have $e\in\OPT(X)$. Indeed, since $e\in\OPT(E)$, the elements preceding $e$
in the greedy order, with ties broken by the fixed rule, do not
span $e$. Restricting those elements to $X$ cannot make them span $e$.

Fix $i>k$ and condition on $e$ arriving in position $i$ and on $X$
being the set of the first $i$ arrivals. Then $e$ arrives last among
the elements of $X$ and $e\in\OPT(X)$. By
Lemma~\ref{lem:algorithm-properties}\ref{item:acceptance-probability},
the conditional probability that $e$ is accepted is $k/(i-1)$. 
Since the value $k/(i-1)$ does not depend on $X$, we may drop the
conditioning on $X$. Thus, conditional on $e$ arriving in position
$i$, it is accepted with probability $k/(i-1)$. Since the arrival
position of $e$ is uniform,
\begin{equation}\label{eq:ratio}
    \Prb\brk{e\in\ALG}
    =
    \frac{1}{n}\sum_{i=k+1}^n\frac{k}{i-1}
    =
    \frac{k}{n}\sum_{j=k}^{n-1}\frac{1}{j}.
\end{equation}
Since $1/x$ is decreasing,
\[
    k\sum_{j=k}^{n-1}\frac{1}{j}
    =
    1+k\sum_{j=k+1}^{n-1}\frac{1}{j}
    \geq
    1+k\int_{k+1}^{n}\frac{1}{x}\dif x
    =
    1+k\ln\frac{n}{k+1}.
\]

It remains to show that the last expression is at least $n/e$.  Since
$k=\floor{n/e}$, we have $ek\leq n<e(k+1)$. Let
$h(x)=1+k\ln(x/(k+1))-x/e$. On $[ek,e(k+1)]$ we have
$h'(x)=k/x-1/e\leq0$, while $h(e(k+1))=0$. Hence $h(n)\geq0$, and thus
$1+k\ln(n/(k+1))\geq n/e$. Combining this with
\eqref{eq:ratio} gives $\Prb\brk{e\in\ALG}\geq1/e$. This completes the proof of Theorem~\ref{thm:main}.

\section{Feasibility via Uncrossing}
\label{sec:feasibility}

We prove Lemma~\ref{lem:lp-feasible} in this section. Throughout, we fix a set of elements $X$ with $|X|=i>k$, together with
a family of distributions
\[
    \boldsymbol{\mu}
    = \bigl(\mu_{X-e}\bigr)_{e\in X},
\]
where each $\mu_{X-e}$ satisfies the Dimension Invariant Inequality
\eqref{eq:dimension-invariant}. Our goal is to prove that $\mathrm{LP}(X)$ is feasible, thereby establishing Lemma~\ref{lem:lp-feasible}.
Let $P=P(X,\boldsymbol{\mu})$ be the polytope obtained from
$\mathrm{LP}(X)$ by dropping the constraints $g_p(L)\leq 0$.
Since $X$ and $\boldsymbol{\mu}$ remain fixed throughout this section,
we simply write $P$.

We first show that $P$ is non-empty. We then show
that there exists a point $p\in P$ that also satisfies
$g_p(L)\leq 0$ for every $L\leq\spn{X}$. For $e\in\OPT(X)$, we have $\dim(\spn{e})=1$.
Applying \eqref{eq:dimension-invariant} to $\mu_{X-e}$
with $L=\spn{e}$ gives
\begin{equation}\label{eq:enough-acceptance-mass}
    \Prb_{W\sim\mu_{X-e}}\brk{e\notin W}
    \geq \frac{k}{i-1}.
\end{equation}
For such an $e$, set $p(e,W)=0$ when $e\in W$, and otherwise set
\[
    p(e,W)
    =
    \frac{k}{
        (i-1)\Prb_{W'\sim\mu_{X-e}}\brk{e\notin W'}
    }.
\]
The inequality above ensures that $p(e,W)\leq 1$, while by construction
$\E_{W\sim\mu_{X-e}}\brk{p(e,W)}=k/(i-1)$. For
$e\notin\OPT(X)$, set $p(e,W)=0$ for every $W$. Thus $p\in P$, so
$P$ is non-empty.

It remains to find a point $p\in P$ such that $g_p(L)\leq 0$ for every
$L\leq\spn{X}$. We begin with a structural property of $g_p$.
Fix $p\in P$ and consider the following one-step experiment. Choose
$e\in X$ uniformly, sample $W$ according to $\mu_{X-e}$, and accept
$e$ with probability $p(e,W)$. Let $Z$ be the resulting spanned subspace.
By Lemma~\ref{lem:dimension-update}, applied with $Y=X$, we have
\[
    g_p(L)
    =
    -(i-k)\dim(L)
    + i\,\E_p\brk{\dim(Z\cap L)}.
\]
We claim that $g_p$ is supermodular on the lattice of subspaces.
Namely, for any $L,H\leq\spn{X}$,
\[
    g_p(L)+g_p(H)
    \leq
    g_p(L\cap H)+g_p(L+H).
\]
Indeed, for every fixed subspace $Z$,
\[
    \dim(Z\cap L)+\dim(Z\cap H)
    =
    \dim((Z\cap L)+(Z\cap H))
    +\dim(Z\cap L\cap H).
\]
Together with
\[
    (Z\cap L)+(Z\cap H)\leq Z\cap(L+H),
\]
this shows that $L\mapsto\dim(Z\cap L)$ is supermodular. Taking
expectations preserves supermodularity, while $\dim(L)$ is modular.
Therefore, $g_p$ is supermodular.

We now use this supermodularity to reduce the problem to chains of
subspaces.
\begin{lemma}\label{lem:chains}
Suppose that for every chain
\[
    \mathcal{C}=\{L_1,\dots,L_s\},
    \qquad
    L_1 < \cdots < L_s \leq \spn{X},
\]
there exists a $p\in P$ such that $g_p(L_j)\leq 0$ for every
$j\in\{1,\dots,s\}$. Then there exists a single $p\in P$ such that
$g_p(L)\leq 0$ for every $L\leq\spn{X}$.
\end{lemma}
\begin{proof}
Since the field is finite, $\spn{X}$ has finitely many subspaces. Let $\Lambda$ be the probability simplex on these subspaces. By the von Neumann minimax theorem,
\begin{equation}\label{eq:dual}
    v \coloneqq \min_{p \in P} \max_{L \leq \spn{X}} g_p(L) = \max_{\lambda \in \Lambda} \min_{ p\in P} \sum_L \lambda_L g_p(L).
\end{equation}
Thus, the existence of a $p \in P$ such that $g_p(L) \leq 0$ for every $L \leq \spn{X}$ is equivalent to $v \leq 0$.

Since $P$ is a polytope and $g_p(L)$ is affine
in $p$, the minimum over $P$ is attained at a vertex.
Thus the map
$\lambda\mapsto\min_{p\in P}\sum_L\lambda_Lg_p(L)$
is the minimum of finitely many linear functions and is
continuous. Its maximizers over the compact simplex
$\Lambda$ form a nonempty compact set.
Choose $\lambda$ among these maximizers to maximize
$\sum_L\lambda_L(\dim(L))^2$. Suppose that two subspaces $L,H$, neither containing the other, have positive mass in $\lambda$, and let $t = \min\set{\lambda_L, \lambda_H}$. Move mass $t$ from each of $L$ and $H$ to $L \cap H$ and $L + H$ and let $\lambda'$ denote the resulting point in $\Lambda$. By the supermodularity of $g_p$, we have $g_p(L) + g_p(H) \leq g_p(L \cap H) + g_p(L + H)$ for every $p \in P$. Hence the weighted sum does not decrease for any $p$, and neither does its minimum over $P$. Since the original $\lambda$ is optimal, the new point $\lambda'$ is also optimal.

On the other hand, $\sum_L \lambda_L (\dim(L))^2$ increases by $2t (\dim(L) - \dim(L \cap H)) (\dim(H) - \dim(L \cap H)) > 0$, a contradiction. Therefore, the chosen optimal $\lambda$ is supported on a chain $\mathcal{C}$. By the hypothesis of the lemma, there exists a $p \in P$ with $g_p(L) \leq 0$ for every $L \in \mathcal{C}$. Thus, $v \leq \sum_L \lambda_L g_p(L) \leq 0$.
\end{proof}

By Lemma~\ref{lem:chains}, it is enough to fix an arbitrary chain $
    L_1 < \cdots < L_s \leq \spn{X}$
and construct a point $p\in P$ such that $g_p(L_j)\leq 0$ for every
$j\in\{1,\dots,s\}$. We will embark on precisely that for the remainder of this section.

For an element $e\in\OPT(X)$, we define
\[
    \mathcal{F}_e
    =
    \set{W\in\supp(\mu_{X-e})\midd e\notin W}.
\]
These are precisely the states in which $e$ can be accepted. We will need the following property.

\begin{lemma}\label{lem:delta-equivalence}
Let $e$ be an element in $\OPT(X)$. For every $W\in\mathcal{F}_e$ and every subspace $L$, we have $$\Delta_e(W,L)=\1{e\in W+L}.$$
\end{lemma}
\begin{proof}
Recall that
$\Delta_e(W,L)=\dim((W+\spn{e})\cap L)-\dim(W\cap L)$.
Use $\dim(A\cap B)=\dim(A)+\dim(B)-\dim(A+B)$,
first with $(A,B)=(W+\spn{e},L)$ and then with
$(A,B)=(W,L)$:
\begin{align*}
    \dim( (W + \spn{e}) \cap L) - \dim(W \cap L) &= \prn{\dim(W + \spn{e}) + \dim(L) - \dim(W + \spn{e} + L)} \\
    &- \prn{\dim(W) + \dim(L) - \dim(W + L)} \\
    &= 1 - \prn{\dim(W + \spn{e} + L) - \dim(W + L)} \\
    &= 1 - \1{e \notin W + L} \\
    &= \1{e \in W + L},
\end{align*}
where the second equality follows from the fact that, since $W \in \mathcal{F}_e$, we have $e \notin W$ and thus $\dim(W + \spn{e}) - \dim(W) = 1$.
\end{proof}

Now we fix $e \in \OPT(X)$, and design the acceptance probabilities $p(e,.)$. For each $j$, let
\[
\mathcal H_j
=\set{W\in\mathcal F_e\midd e\notin W+L_j}.
\]
By Lemma~\ref{lem:delta-equivalence}, accepting $e$
does not increase $\dim(W\cap L_j)$ when $W\in\mathcal H_j$,
and increases it by one when
$W\in\mathcal F_e\setminus\mathcal H_j$.
Since $L_1<\cdots<L_s$, we have
$\mathcal H_s\subseteq\cdots\subseteq\mathcal H_1$. Intuitively, $\mathcal H_s$ consists of the ``safest'' states $W$ for accepting $e$, in that doing so does not increase $\dim(W\cap L_j)$ for any subspace $L_j$ in our chain.  More generally, for $1\leq j<s$, the states in
$\mathcal H_j\setminus\mathcal H_{j+1}$ are safe for
$L_1,\dots,L_j$ but unsafe for $L_{j+1},\dots,L_s$.
Enumerate $\mathcal F_e$ as $W_1,\dots,W_t$, in order of most to least safe, by listing
first the states in $\mathcal H_s$, then those in
$\mathcal H_{s-1}\setminus\mathcal H_s$, and so on,
ending with $\mathcal F_e\setminus\mathcal H_1$ which are unsafe for all $L_j$.
The order within each block is arbitrary.
For every $j$, the states in $\mathcal H_j$ therefore
form an initial segment of this enumeration.
For $W\in\supp(\mu_{X-e})$ with $e\in W$, set $p(e,W)=0$.
For $\ell=1,\dots,t$, assign acceptance probabilities
in that order:
\begin{equation}\label{eq:chain-solution}
p(e,W_\ell)
=
\min\set{
1,
\frac{
\brk{\frac{k}{i-1}
-\sum_{h<\ell}\mu_{X-e}(W_h)}^{+}
}{\mu_{X-e}(W_\ell)}
},
\end{equation}
where $[x]^+\coloneqq\max\set{x,0}$.
Since each $W_\ell$ has positive probability under
$\mu_{X-e}$, the denominator of the expression above is positive.

Recall that $e\in\OPT(X)$ is fixed throughout the above construction of the probabilities $p(e,.)$. By \eqref{eq:enough-acceptance-mass} and the definition of
$\mathcal{F}_e$,
\[
    \mu_{X-e}(\mathcal{F}_e)
    =
    \Prb_{W\sim\mu_{X-e}}\brk{e\notin W}
    \geq \frac{k}{i-1}.
\]
Hence the greedy assignment in \eqref{eq:chain-solution} assigns total
acceptance probability exactly $k/(i-1)$ to $e$.
We perform the construction above separately for each $e\in\OPT(X)$,
and set $p(e,W)=0$ for every $e\notin\OPT(X)$. This defines all the
coordinates of a point $p$. Since, for each $e\in\OPT(X)$, the
construction assigns total acceptance probability $k/(i-1)$ and sets
$p(e,W)=0$ whenever $e\in W$, we have $p\in P$.

Now fix $e\in\OPT(X)$. Since every $\mathcal{H}_j$ is an initial segment,
\begin{equation}\label{eq:initial-segment-mass}
    \sum_{W \in \mathcal{H}_j} \mu_{X - e}(W) p(e, W)
    =
    \min\set{\frac{k}{i - 1}, \mu_{X - e}(\mathcal{H}_j)}.
\end{equation}
For every $L$ in the chain, let
$q_e^p(L)=\E_{W\sim\mu_{X-e}}[p(e,W)\Delta_e(W,L)]$
be the expected increase in $\dim(W\cap L)$ when we
accept $e$ according to $p$.
Recall that $u_e(L)=\E_{W\sim\mu_{X-e}}[\dim(W\cap L)]$,
so $u_e(L)+q_e^p(L)$ is the expected dimension of
the intersection of $L$ with the span of the accepted
elements after processing $e$. We will establish an upper bound on $u_e(L)+q_e^p(L)$, before summing over all $e$ to show that $g_p(L) \leq 0$. First, we re-express $q^p_e(L)$ as follows.

\begin{align}
    q_e^p(L) &= \sum_{W \in \mathcal{F}_e} \mu_{X - e}(W) p(e, W) \Delta_e(W, L) \nonumber \\
    &= \sum_{W \in \mathcal{F}_e : e \in W + L} \mu_{X - e}(W) p(e, W) \nonumber \\
    &= \sum_{W \in \mathcal{F}_e}\mu_{X - e}(W) p(e, W)
    - \sum_{W \in \mathcal{F}_e : e \notin W + L}\mu_{X - e}(W) p(e, W) \nonumber \\
    &= \frac{k}{i - 1}
    - \min\set{\frac{k}{i-1},
    \mu_{X - e}\prn{\set{W \in \mathcal{F}_e \midd e \notin W + L}} } \nonumber \\
    &= \prn{\frac{k}{i - 1}
    - \mu_{X - e}\prn{\set{W \in \mathcal{F}_e \midd e \notin W + L}} }^+ \nonumber \\
    \label{eq:qepl}
    &= \prn{\frac{k}{i - 1}
    - \Prb_{W \sim \mu_{X - e}} \brk{e \notin W + L}}^+,
\end{align}
where the second equality follows from Lemma~\ref{lem:delta-equivalence},
and the fourth follows from \eqref{eq:initial-segment-mass}.

Suppose first that $e\notin L$. Similarly to Lemma~\ref{lem:delta-equivalence}, the modularity of dimension on subspaces easily yields
\[
\dim(W\cap(L+\spn{e}))-\dim(W\cap L)
=\1{e\in W+L}.
\]
Taking expectations and rearranging, we obtain
\begin{equation}\label{eq:larger-dim}
u_e(L)+\frac{k}{i-1}
-\Prb_{W\sim\mu_{X-e}}\brk{e\notin W+L}
=
u_e(L+\spn{e})-\left(1-\frac{k}{i-1}\right).
\end{equation}
Combining \eqref{eq:qepl} and \eqref{eq:larger-dim} gives
\begin{align}
u_e(L)+q_e^p(L)
&=u_e(L)+
\brk{\frac{k}{i-1}
-\Prb_{W\sim\mu_{X-e}}\brk{e\notin W+L}}^+
\nonumber\\
&=\max\set{u_e(L),
u_e(L+\spn{e})-\left(1-\frac{k}{i-1}\right)}.
\label{eq:slack-bound}
\end{align}
The Dimension Invariant Inequality \eqref{eq:dimension-invariant} for $X - e$, when applied to $L$ and $L + \spn{e}$, gives
\begin{align*}
u_e(L)
&\leq \prn{1 - \frac{k}{i - 1}} \dim(L),\\
u_e(L + \spn{e})
&\leq \prn{1 - \frac{k}{i - 1}} \dim(L + \spn{e})
= \prn{1 - \frac{k}{i - 1}} (\dim(L) + 1).
\end{align*}
Combining this with \eqref{eq:slack-bound} yields
\begin{equation}\label{eq:max}
    u_e(L) + q_e^p(L) = \max\set{u_e(L), u_e(L + \spn{e}) - \prn{1 - \frac{k}{i - 1}}} \leq \prn{1 - \frac{k}{i - 1}} \dim(L).
\end{equation}
If $e\in L$, Lemma~\ref{lem:delta-equivalence} gives
$\Delta_e(W,L)=1$ for every $W\in\mathcal F_e$.
Thus $q_e^p(L)=k/(i-1)$, and \eqref{eq:dimension-invariant}
gives $u_e(L)\leq(1-\frac{k}{i-1})\dim(L)$.
Together with \eqref{eq:max} for $e\notin L$, this gives
\[
    u_e(L) + q_e^p(L) \leq \prn{1 - \frac{k}{i - 1}} \dim(L) + \frac{k}{i - 1} \cdot \1{e \in L}.
\]
For $e\notin\OPT(X)$, we have $p(e,W)=0$ for every $W$, and the Dimension Invariant Inequality gives
\[
    u_e(L) \leq \prn{1 - \frac{k}{i - 1}} \dim(L).
\]
Since $\OPT(X)$ is independent,
$|\OPT(X) \cap L| \leq \dim(L)$. Summing these bounds over all elements in $X$ gives
\begin{align*}
    \sum_{e \in X} u_e(L) + \sum_{e' \in \OPT(X)} q_{e'}^p(L) &\leq i \; \prn{1 - \frac{k}{i - 1}} \dim(L) + \frac{k}{i - 1} \cdot |\OPT(X) \cap L| \\
    &\leq \prn{i \; \prn{1 - \frac{k}{i - 1}} + \frac{k}{i - 1}} \dim(L) = (i - k) \dim(L).
\end{align*}
Therefore $g_p(L) \leq 0$ for every $L$ in the chain. Lemma~\ref{lem:chains} gives a single $p \in P$ satisfying all the constraints of $\mathrm{LP}(X)$. This proves
Lemma~\ref{lem:lp-feasible}.

\begin{remark}\label{rem:arbitrary-fields}
Ignoring computational issues, Theorem~\ref{thm:main} also holds
in the online-representation model over arbitrary fields,
including infinite fields. Although the observed span is
finite-dimensional, it may then have infinitely many subspaces.
We show that the LP still has finitely many variables and
distinct constraints, and that the proof of
Lemma~\ref{lem:chains} continues to apply.

Fix an observed set $Y$ with $|Y|>k$, and let
$\mathcal S_Y=\set{\spn{Z}\colon Z\subseteq Y}$, which is finite.
By induction, $\supp(\mu_{Y-e})\subseteq\mathcal S_Y$ for every
$e\in Y$. Thus there are finitely many variables $p_Y(e,W)$.
Moreover, $W+\spn{e}\in\mathcal S_Y$ for every such pair.
The constraint indexed by $L\leq\spn{Y}$ therefore depends
only on $\dim(L)$ and its intersection dimensions with
the subspaces in $\mathcal S_Y$. We record this information
in the signature
\[
    \sigma_Y(L)=
    \bigl(\dim(L),\,(\dim(W\cap L))_{W\in\mathcal S_Y}\bigr).
\]
Each entry is an integer between $0$ and $|Y|$, so only finitely
many signatures occur. Subspaces with the same signature
define the same constraint. We can therefore retain one
representative subspace for each signature, obtaining a finite LP.

The proof of Lemma~\ref{lem:chains} now applies to these
representatives. In its uncrossing step, suppose that two
chosen representatives $L,H$ have positive $\lambda$-mass
and neither contains the other. We move equal amounts
of mass from $L,H$ to the chosen representatives with signatures
$\sigma_Y(L\cap H)$ and $\sigma_Y(L+H)$, respectively.
Replacing a subspace with another having the same signature
preserves both its LP constraint and its dimension.
As in the proof of Lemma~\ref{lem:chains}, the minimax objective
does not decrease, while $\sum_L\lambda_L(\dim(L))^2$
strictly increases.
The construction in~\eqref{eq:chain-solution} also works without change. The coefficients of the resulting finite LP are still rational, so Lemma~\ref{lem:algorithm-properties}\ref{item:rationality} continues to hold, and in Algorithm~\ref{alg:update} we enumerate one representative subspace per signature instead of all subspaces of~$\spn{Y}$.
\end{remark}

\section{Finitary Modular Extensions}
\label{sec:extensions}

We now turn to the known-matroid model. For linear matroids, the result follows from Theorem~\ref{thm:main}, since every finite linear matroid admits a representation over a finite field. We then show that the same guarantee holds more generally for matroids admitting a finitary modular extension.

\begin{corollary}\label{cor:known-linear}
In the known-matroid model, every linear matroid admits an ordinal $1/e$-probability-competitive algorithm.
\end{corollary}

\begin{proof}
By a theorem of Rado \cite{Rado1957}, every finite matroid representable over some field is representable over a finite field. Since the matroid is known and we allow unbounded computational power, we may enumerate finite fields and matrices with columns indexed by $E$, checking each candidate against the independence structure of the matroid, until we find such a representation. We do not require this search to be efficient. Once a finite-field representation is fixed, Theorem~\ref{thm:main} applies.
\end{proof}

Thus the representation need not be supplied as part of the input in the known-matroid model. In particular, the corollary applies to every finite linear matroid, regardless of the field over which it is initially known to be representable.

For a matroid $N$, we write $\clo_N$ for its \emph{closure operator}. A set $F$ is a \emph{flat} if $\clo_N(F)=F$. For two flats $F_1,F_2$, their \emph{meet} is $F_1\cap F_2$ and their \emph{join} is $F_1\vee F_2\coloneqq \clo_N(F_1\cup F_2)$. The pair $F_1,F_2$ is \emph{modular} if $r_N(F_1) + r_N(F_2) = r_N(F_1\cap F_2) + r_N(F_1\vee F_2)$. The matroid is \emph{modular} if every pair of its flats is modular. A possibly infinite matroid is \emph{finitary} if every circuit is finite. We say that a finite matroid $M$ on $E$ admits a \emph{finitary modular extension} if there is a finitary modular matroid $M'$ on a ground set containing $E$ such that $M'|E = M$. We only use the restriction of $M'$ to
$\clo_{M'}(E)$, which has finite rank and is again modular.
Every circuit of this restriction has at most $r_{M'}(E)+1$
elements, so the restriction is finitary.

The proof of Theorem~\ref{thm:main} in fact uses less than a vector-space representation. The arguments only require the modular rank identities for flats. This yields the following more general consequence.

\begin{corollary}\label{cor:modular-extension}
Let $M$ be a known matroid admitting a finitary modular extension. Then $M$ admits an ordinal $1/e$-probability-competitive algorithm.
\end{corollary}

\begin{proof}
Suppose first that an extension $M'$ is fixed and available to the
algorithm. We work with the flats of $M'$ contained in $\clo_{M'}(E)$, ordered by inclusion. Since $E$ is finite, each such flat has rank at most $r_{M'}(E)<\infty$.

The construction of Section~\ref{sec:algorithm} carries over by replacing
vector spans with closures in $M'$, sums of subspaces with joins of flats,
and dimension with $r_{M'}$. Thus, the state corresponding to an accepted
set $A$ is $\clo_{M'}(A)$, the initial state is
$\clo_{M'}(\emptyset)$, and the Dimension Invariant Inequality becomes
\[
    \E_{W\sim\mu_X} \brk{r_{M'}(W \cap L)} \leq \prn{1 - \frac{k}{|X|}} r_{M'}(L)
\]
for every flat $L \subseteq \clo_{M'}(E)$.

For every fixed flat $W$, the function
$F\mapsto r_{M'}(W\cap F)$ is supermodular. Indeed,
$(W\cap F_1)\vee(W\cap F_2)\subseteq W\cap(F_1\vee F_2)$,
and we can apply modularity to the pairs
$W\cap F_1,W\cap F_2$ and $F_1,F_2$.
We can therefore repeat the uncrossing argument from
Section~\ref{sec:feasibility}.

For the construction along a fixed chain, let $L$ be a
flat in the chain. Modularity gives
\begin{align*}
r_{M'}(\clo_{M'}(W\cup\{e\})\cap L)
&=r_{M'}(W\cap L)+\1{e\in\clo_{M'}(W\cup L)}
&&\text{if }e\notin W,\\
r_{M'}(W\cap\clo_{M'}(L\cup\{e\}))
&=r_{M'}(W\cap L)+\1{e\in\clo_{M'}(W\cup L)}
&&\text{if }e\notin L.
\end{align*}
Also, $\OPT(X)$ is independent in $M=M'|E$, so
$|\OPT(X)\cap L|\leq r_{M'}(L)$.
We can now repeat the construction along a fixed chain and the
proof of \eqref{eq:max}, using closures and rank in place of
spans and dimension.

There may be infinitely many flats contained in $\clo_{M'}(E)$.
As in Remark~\ref{rem:arbitrary-fields}, we use signatures
to show that the acceptance LP nevertheless has only finitely
many distinct constraints. For a fixed observed set $X$, every state that can occur is of the form $\clo_{M'}(Z)$ for some $Z \subseteq X$.
Hence the set
\[
    \mathcal{S}_X = \set{\clo_{M'}(Z)\colon Z\subseteq X}
\]
is finite. For a flat $L\subseteq\clo_{M'}(E)$, define its signature relative to $X$ by
\[
    \sigma_X(L) = \prn{r_{M'}(L), \prn{r_{M'}(W \cap L)}_{W \in \mathcal{S}_X} }.
\]
Every quantity appearing in the constraint $g_p(L) \leq 0$ is determined by this signature. Indeed, both a predecessor state $W$ and the state obtained after accepting an element belong to $\mathcal S_X$. Since all ranks in $\sigma_X(L)$ are integers between $0$ and $r_{M'}(E)$, only finitely many signatures can occur. Thus there are only finitely many distinct constraints.

The chain-reduction argument can now be applied to these finitely many constraint types. Choose one representative flat for each signature. If two representatives carrying positive mass in an optimal minimax solution are incomparable, move their mass, exactly as in the proof of Lemma~\ref{lem:chains}, to the chosen representatives of the signatures of their intersection and join. Flats with the same signature define the same constraint and have the same rank, so supermodularity ensures that the minimax objective does not decrease, while modularity gives the same strict increase in the secondary objective. Therefore an optimal solution can again be supported on a chain. The construction for a fixed chain uses only the Dimension Invariant Inequality and the properties established above, so the remainder of the proof of Lemma~\ref{lem:lp-feasible} carries over unchanged. We therefore obtain an ordinal randomized policy with the same $1/e$ probability guarantee as in Theorem~\ref{thm:main}.

This policy was constructed using the extension $M'$. Finally, we show that the extension need not be given to the algorithm. Since $M$ is finite, there are only finitely many deterministic ordinal online policies that always maintain independence in $M$. A history is determined by the ordered sequence of labels seen so far, their relative weight order, and the previous accept/reject decisions. Fixing all random choices of the policy constructed above in advance therefore induces a probability distribution over this finite set of deterministic policies.

We can recover a suitable distribution using only the known matroid $M$. Consider the finite LP that chooses a probability distribution over deterministic ordinal policies and maximizes a variable $t$. For every total order on $E$ and every element of the corresponding optimal basis, require that its selection probability, averaged over all arrival orders, be at least $t$. All coefficients of this LP are determined by $M$. The randomized policy constructed from $M'$ provides a feasible distribution with $t \geq 1/e$, so the optimum of the LP is at least $1/e$. We may therefore solve this LP before the arrivals, sample a deterministic policy from the resulting distribution, and follow it online. The resulting algorithm depends only on the known matroid $M$ and has the required probability guarantee.
\end{proof}

For finite matroids, admitting a finitary modular extension
is equivalent to being fully modular extendable.
The existence of such an extension for every fully modular
extendable matroid follows from
\cite[Lemma~3.2]{BercziEtAl}.
Hence Corollary~\ref{cor:modular-extension} applies precisely
to fully modular extendable matroids.
By \cite[Theorem~3.4]{BercziEtAl}, these are exactly the matroids
whose connected components either have rank $3$
or are skew-representable.
In particular, Corollary~\ref{cor:modular-extension} covers
every matroid of rank at most $3$, since matroids of rank
at most $2$ are linear.
Its scope strictly extends beyond linear matroids:
the rank-$3$ non-Pappus and non-Desargues matroids
are not linear, but both satisfy its hypotheses.
By contrast, the V\'amos matroid is connected, has rank $4$,
and is not skew-representable, so it does not admit
a finitary modular extension.

\section{XOS Valuations and Group Arrivals}
\label{sec:xos}

A valuation $f:2^E\to\mathbb R_{\geq 0}$ is \emph{XOS}
(\emph{fractionally subadditive})~\cite{LehmannLN06,Feige09} if there exists a finite nonempty family
$\mathcal Q\subseteq\mathbb R_{\geq 0}^E$ of nonnegative additive
valuations such that, for all $D\subseteq E$,
\[
f(D)=\max_{q\in\mathcal Q}\sum_{e\in D}q_e.
\]
Each $q\in\mathcal Q$ is an additive \emph{clause}. 
The following simple observation will be useful in our analysis. For any
$B\subseteq E$, fix a \emph{supporting} clause $q$ with
$f(B)=\sum_{e\in B}q_e$; then $f(D)\geq \sum_{e\in D}q_e$ for all
$D\subseteq E$. Consequently, if $T\subseteq B$ is random and
$\Prb[e\in T]=\alpha$ for all $e\in B$, then
\begin{equation}
\E[f(T)]\geq \E\left[\sum_{e\in T}q_e\right]=\alpha f(B).
\label{eq:xos-fractional}
\end{equation}

Observe that additive valuations are XOS with $|\mathcal Q|=1$.
A unit-demand valuation is also XOS: for nonnegative weights
$(w_e)_{e\in E}$, define $f(\emptyset)=0$ and
$f(D)=\max_{e\in D}w_e$ for $D\neq\emptyset$.
For $E\neq\emptyset$, take $\mathcal Q$ to consist of
the clauses that assign weight $w_e$ to one element $e$
and $0$ to all other elements.
If $E=\emptyset$, take $\mathcal Q$ to consist of the
empty clause.

\paragraph{XOS Matroid Secretary. }
Let $M=(E,\mathcal I)$ be a known matroid, and let $N$ be a set
of $n$ agents, where $n$ is known in advance. The agents arrive
in uniformly random order. Each agent $a$ reveals an eligibility
set $E_a\subseteq E$ and an XOS valuation
$f_a:2^{E_a}\to\mathbb R_{\geq 0}$. Upon arrival, the algorithm
irrevocably assigns a bundle $T_a\subseteq E_a$ to agent $a$. Eligibility
sets may overlap, but the assigned bundles must be pairwise
disjoint, and their union must be independent in $M$. The goal is
to maximize $\sum_{a\in N} f_a(T_a)$. We remark that all eligibility sets and
valuations are fixed before the arrival order is drawn, and we allow unbounded computation and value queries.

For a set of agents $X\subseteq N$, we represent an allocation
by a set $S=\{(a,T_a):a\in X\}$, where $T_a\subseteq E_a$.
Write
\[
\overline{S}=\bigcup_{a\in X}T_a
\qquad\text{and}\qquad
F(S)=\sum_{a\in X}f_a(T_a)
\]
for its set of assigned elements and its value. Throughout this
section, $\OPT(X)$ and $\ALG(X)$ denote allocations for the offline optimal solution and the algorithm, respectively. In particular,
$\ALG(X)=\{(a,T_a):a\in X\}$ after the agents of $X$ have been
processed, including pairs with empty bundles.
Once agent $a$ has arrived, the algorithm keeps $E_a$
and can query $f_a$ throughout the execution.
For every set $X$ of observed agents, let
\[
\OPT(X)=\{(a,B_a(X)):a\in X\}
\]
be an offline optimal allocation for the sub-instance on $X$, chosen
by a fixed deterministic tie-breaking rule. We compute $\OPT(X)$
over all feasible allocations for the agents in $X$, using their
eligibility sets and valuations.
Its value is $F(\OPT(X))=\sum_{a\in X}f_a(B_a(X))$.

When $M$ is the free matroid on $E$, we only require the
assigned bundles to be pairwise disjoint. We then obtain the
online combinatorial auction problem, for which Kesselheim,
Radke, T\"onnis, and V\"ocking~\cite{KesselheimRTV13} obtained an optimal
$1/e$ guarantee for submodular
valuations.\footnote{Their proof uses Proposition~8 of
\cite{KesselheimRTV13} to bound the expected value of a random subset.
The supporting-clause argument in \eqref{eq:xos-fractional}
proves the same bound for XOS valuations, also when each
inclusion probability is at least $\alpha$.
Their analysis therefore extends to XOS valuations.}
Our Theorem~\ref{thm:xos} generalizes their guarantee from the free
matroid to every known linear matroid.
The model also includes the following two settings: \emph{group matroid secretary} and \emph{agent matroid secretary}, in which each agent receives at most one element.

\paragraph{Group Matroid Secretary. }
An adversary partitions $E$ into $n$ groups $(G_a)_{a\in N}$.
The groups arrive in uniformly random order; $n$ is known in
advance, but the partition is not. When a group arrives, its
elements and their nonnegative weights are revealed. The
algorithm may accept any subset of the arriving group, as long
as the union of all accepted elements remains independent in
$M$. Korula and P\'al~\cite{korula-pal-graphic} introduced
secretary problems with groups and studied this model for
matching.

This is a special case of the XOS model: take $E_a=G_a$ and
$f_a(D)=\sum_{e\in D} w_e$ for every $D\subseteq G_a$. 
Order $E$ by nonincreasing weight, using the fixed
tie-breaking rule from Section~\ref{sec:intro}. For every
observed set of groups $X\subseteq N$, let $I(X)$ be the greedy
maximum-weight basis  in $\bigcup_{a\in X} G_a$ under
this order, and set $B_a(X)=I(X)\cap G_a$. We choose
$\OPT(X)=\{(a,B_a(X)):a\in X\}$. Then
$\overline{\OPT(X)}=I(X)$ and $F(\OPT(X))=w(I(X))$.

\paragraph{Agent Matroid Secretary. }
Each agent $a$ reveals its eligible elements $E_a$ and
nonnegative weights $(w_{a,e})_{e\in E_a}$. At most one element
may be assigned to each agent, each element may be assigned at
most once, and the set of assigned elements must be independent
in $M$. Define $f_a(\emptyset)=0$ and
$f_a(D)=\max_{e\in D} w_{a,e}$ for every nonempty
$D\subseteq E_a$. In any feasible allocation, we can replace
each nonempty bundle by one of its maximum-weight elements
without changing the total value or violating feasibility.
We therefore choose each $\OPT(X)$ so that
$|B_a(X)|\leq 1$ for every $a\in X$.

A special case is weighted bipartite matching with one-sided
vertex arrivals: offline vertices are known in advance, online
vertices arrive in random order, and reveal their incident edges
and weights. Each online vertex must be matched immediately to
an available offline neighbor or rejected. This corresponds to
taking $M$ to be the free matroid on the offline vertices.
Kesselheim, Radke, T\"onnis, and V\"ocking~\cite{KesselheimRTV13}
gave a $1/e$-competitive algorithm for this problem.
Marinkovic, Soto, and Verdugo~\cite{marinkovic-soto-verdugo}
extended this online assignment model to general independence
systems on the items, including hypergraph matchings and
matroid intersections admitting directed certifiers.

\paragraph{Approximation guarantees. } The following theorem is the main result of this section.
\begin{theorem}
\label{thm:xos}
For every known linear matroid, the XOS Matroid Secretary
Problem admits an online algorithm that always produces a
feasible allocation and satisfies
\[
\E[F(\ALG)]\geq \frac{1}{e}F(\OPT(N)).
\]
The same guarantee holds without prior knowledge of the matroid when a representation over a known finite field is revealed with the agents, in a common coordinate system.
\end{theorem}
The algorithm used to prove Theorem~\ref{thm:xos}
also gives the following corollaries for group and agent arrivals.
\begin{corollary}
\label{cor:xos-groups}
For Group Matroid Secretary on a known linear matroid, there
is an online algorithm that accepts each element of a fixed
greedy optimum with probability at least $1/e$.
The same guarantee holds without prior knowledge of the matroid when a representation over a known finite field is revealed with the groups, in a common coordinate system.
\end{corollary}

\begin{corollary}
\label{cor:xos-agents}
For Agent Matroid Secretary on a known linear matroid, there
is an online algorithm whose expected value is at least $1/e$
times the optimum.
\end{corollary}

To prove Theorem~\ref{thm:xos}, we extend Algorithm~\ref{alg:online} to agent arrivals.
In that algorithm, we decide whether to accept or reject
the arriving element. Here, when agent $a$ arrives, we
assign it a subset of its \emph{optimal} bundle $B_a(X)$ in $\OPT(X)$,
where $X$ is the set of agents observed so far.
In \Cref{subsec:xos-invariant}, we establish the Dimension Invariant Inequality we need in the XOS matroid secretary context.
Then, in \Cref{subsec:xos-algorithm}, we show how to construct the linear program that encodes the distribution we use to sample at each iteration of the selection algorithm, and in \Cref{subsec:xos-feasibility} we show its feasibility via an uncrossing argument.
Finally, in \Cref{subsec:xos-guarantees}, we put things together to conclude Theorem~\ref{thm:xos}, and then we prove the approximation guarantees in \Cref{cor:xos-groups} and \Cref{cor:xos-agents}.

For Sections~\ref{subsec:xos-invariant}--\ref{subsec:xos-feasibility},
fix a representation of $M$ in a finite-dimensional vector
space $V$ over a finite field. We use the same notation for
an element and its representing vector, but retain the element's
label. Distinct elements may have the same representing vector.
In the proof of Theorem~\ref{thm:xos},
we compute such a representation from the known matroid before
the arrivals, as in Corollary~\ref{cor:known-linear}.

\subsection{The Dimension Invariant for Bundles}
\label{subsec:xos-invariant}

Fix a sample size $k\in \{1,\ldots,n-1\}$. 
For every set of observed agents $X$, let $V_X\coloneqq\spn{\bigcup_{a\in X}E_a}$.
As in Section~\ref{sec:invariant}, we recursively construct
a distribution $\mu_X$ supported on the subspaces of $V_X$.
Conditional on $X$ being the set of the first $|X|$ arrivals,
$\mu_X$ will be the distribution of
$\spn{\overline{\ALG}}$ after processing those agents.

Agents in the sample receive empty bundles, so for $|X|\leq k$, we set $\mu_X=\delta_{\{0\}}$, the point mass
at the zero subspace. For $|X|\geq k$, we require the following \emph{Dimension Invariant Inequality}:
\begin{equation}
\E_{W\sim\mu_X}[\dim(W\cap L)]
\leq \left(1-\frac{k}{|X|}\right)\dim(L), \quad \text{for every $L\leq V$.}
\label{eq:xos-invariant}
\end{equation}
This is \eqref{eq:dimension-invariant} with the number of
observed agents in place of the number of observed elements. As in Section~\ref{sec:invariant}, it suffices to impose
\eqref{eq:xos-invariant} for $L\leq V_X$.
Indeed, every $W\in\supp(\mu_X)$ satisfies $W\cap L=W\cap(L\cap V_X)$ and
$\dim(L\cap V_X)\leq\dim(L)$.

Consider the arrival of agent $a$ in position $i>k$, and
let $X$ be the set of the first $i$ agents. Conditional on
$X$ and $a$, the agents in $X-a$ arrive in uniformly random
order. Let $W=\spn{\overline{\ALG}}$ be the span of the
elements assigned before $a$ arrives, so $W\sim\mu_{X-a}$.

Put $\alpha=k/(i-1)$. We seek a random subset
$R\subseteq B_a(X)$ that includes each element with
probability $\alpha$; recall that $B_a(X)\subseteq E_a$ is such that $(a,B_a(X))\in \OPT(X)$. Even if each element can be added
on its own, we may be unable to add them together.
For example, if $u,v$ are linearly independent and
$W=\spn{u+v}$, adding either $u$ or $v$ increases the
dimension by one. Adding both still increases it by
only one.
To preserve feasibility, we also require that $\overline{\ALG}\cup R$ be
independent and that $R$ contain no previously assigned
element. Since the previously assigned elements are
independent and span $W$, these requirements are
equivalent to
\begin{equation}
\dim(W+\spn{R})=\dim(W)+|R|.
\label{eq:xos-feasible-subset}
\end{equation}
For each $e\in B_a(X)$, the probability $\Prb[e\in R]$
is taken over $W\sim\mu_{X-a}$ and the randomized choice
of $R$ given $W$. Since $f_a$ is XOS, we can apply
\eqref{eq:xos-fractional} to the bundle $B_a(X)$ and its
random subset $R$. Each element belongs to $R$ with
probability $\alpha$, so
\begin{equation}
\E[f_a(R)]\geq\alpha f_a(B_a(X)).
\label{eq:xos-bundle-value}
\end{equation}
We will use \eqref{eq:xos-bundle-value} to bound the
expected value assigned at each arrival in the proof
of Theorem~\ref{thm:xos}.

Since $B_a(X)$ is independent, every
$D\subseteq B_a(X)$ satisfies $\dim(\spn{D})=|D|$.
Applying \eqref{eq:xos-invariant} to $\mu_{X-a}$
with $L=\spn{D}$ therefore gives
\begin{equation}
\E[\dim(W\cap\spn{D})]\leq(1-\alpha)|D|
\qquad\text{for every }D\subseteq B_a(X),
\label{eq:xos-bundle-intersection}
\end{equation}
where we used $|X-a|=i-1$ and $\alpha=k/(i-1)$.
The following lemma uses the bounds in
\eqref{eq:xos-bundle-intersection} to construct $R$
satisfying \eqref{eq:xos-feasible-subset} and
$\Prb[e\in R]=\alpha$ for every $e\in B_a(X)$.

\begin{lemma}
\label{lem:xos-marginals}
Let $\alpha\in[0,1]$, and let $W$ be a random subspace of $V$
such that, for every $L\leq V$,
\begin{equation}
\E[\dim(W\cap L)]
\leq (1-\alpha)\dim(L).
\label{eq:xos-marginal-assumption}
\end{equation}
For every independent set $C$ of vectors, there is a randomized
rule that, after observing $W$, selects $R\subseteq C$ satisfying
$\dim(W+\spn{R})=\dim(W)+|R|$ and
$\Prb[c\in R]=\alpha$ for every $c\in C$.
\end{lemma}

\begin{proof}
Fix an independent set $C$, and let $\nu$ be the
distribution of $W$. For each $W$, let $M_W$ be the matroid on $C$ with
rank function given, for $D\subseteq C$, by
\begin{equation}
r_W(D)=\dim(W+\spn{D})-\dim(W)
=\dim(\spn{D})-\dim(W\cap\spn{D}).
\label{eq:xos-rank}
\end{equation}
To see that $r_W$ is a matroid rank function, introduce
a set $B_W$ of new elements, disjoint from $C$, whose
representing vectors form a basis of $W$. Let $r$ be the
rank function of the linear matroid on $C\cup B_W$.
Contracting $B_W$ leaves ground set $C$ and gives rank
$r(D\cup B_W)-r(B_W)=r_W(D)$ for every $D\subseteq C$.
A set $R\subseteq C$ is independent in $M_W$ exactly
when it satisfies \eqref{eq:xos-feasible-subset}.

Let $P(M_W)$ be the matroid independence polytope, that is,
the convex hull of the incidence vectors of the independent
sets of $M_W$. This is an integral polymatroid with rank
function $r_W$. Define
\[
P=\left\{
\sum_W\nu(W)z^W:\text{ for each $W$, the vector $z^W$ belongs to $P(M_W)$}
\right\}.
\]
For each $z^W\in P(M_W)$, we can select a random independent
set of $M_W$ with inclusion probabilities $z^W$.
Thus $P$ consists of the inclusion probabilities we can
obtain after averaging over $W$. It suffices to show
that $\alpha\mathbf 1_C\in P$.

When we add polymatroids with nonnegative weights,
we add their rank functions with the same weights; see, e.g.,
\cite[Theorem~44.6]{Schrijver2003}.
Thus, for polymatroids $P_i$ on $C$ with rank functions
$r_i$ and weights $\lambda_i\geq0$,
\[
\sum_i\lambda_iP_i
=
\left\{
x\in\mathbb R_{\geq0}^C:
\sum_{c\in D}x_c\leq\sum_i\lambda_i r_i(D)
\quad\forall D\subseteq C
\right\}.
\]
Applying this identity to $P(M_W)$ with weights $\nu(W)$
gives
\begin{equation}
P=
\left\{
x\in\mathbb R_{\geq0}^C:
\sum_{c\in D}x_c\leq\E[r_W(D)]
\quad\forall D\subseteq C
\right\}.
\label{eq:xos-expected-rank-polytope}
\end{equation}
Thus $P$ is a polymatroid with rank function
$\E[r_W(\cdot)]$. This function need not be integer-valued,
so $P$ need not be a matroid polytope. We now check the inequalities in
\eqref{eq:xos-expected-rank-polytope} for
$x=\alpha\mathbf 1_C$.
Fix $D\subseteq C$. Since $C$ is independent,
$\dim(\spn{D})=|D|$.
Apply \eqref{eq:xos-marginal-assumption} with $L=\spn{D}$.
Using \eqref{eq:xos-rank}, we obtain
\begin{align*}
\E[r_W(D)]
&=|D|-\E[\dim(W\cap\spn{D})] \geq |D|-(1-\alpha)|D| =\alpha|D|.
\end{align*}
Since $\sum_{c\in D}x_c=\alpha|D|$, the inequality
in \eqref{eq:xos-expected-rank-polytope} holds for every
$D\subseteq C$. Hence $\alpha\mathbf 1_C\in P$.

For each $W$, choose $z^W\in P(M_W)$ such that
$\sum_W\nu(W)z^W=\alpha\mathbf 1_C$. 
For each $W$, express $z^W$ as a convex combination
of incidence vectors of independent sets of $M_W$.
After observing $W$, choose one of these independent
sets as $R$, with probabilities given by the coefficients.
Then $\dim(W+\spn{R})=\dim(W)+|R|$ and
$\Prb[c\in R\mid W]=z^W_c$ for every $c\in C$.
Hence
\[
\Prb[c\in R]=\sum_W\nu(W)z^W_c=\alpha
\qquad\text{for every }c\in C. \qedhere
\]
\end{proof}

Lemma~\ref{lem:xos-marginals} extends the acceptance rule
of Section~\ref{sec:feasibility} to bundles.
For a single element, \eqref{eq:enough-acceptance-mass}
gives $\Prb[e\notin W]\geq\alpha$, and that rule accepts
$e$ with probability $\alpha$. The lemma gives the
corresponding choice for $B_a(X)$: each element is included
with probability $\alpha$, and the selected elements
satisfy \eqref{eq:xos-feasible-subset} together.
The LP in Section~\ref{subsec:xos-algorithm}
will also impose \eqref{eq:xos-invariant} on the
resulting distribution $\mu_X$.

\subsection{The Allocation LP and the Algorithm}
\label{subsec:xos-algorithm}

We now construct $\mu_X$ recursively from the distributions
for smaller sets of agents. As in Section~\ref{subsec:lp},
we use a linear program to choose the assignment
probabilities at each step. For $|X|\leq k$, we use
$\mu_X=\delta_{\{0\}}$ and store $p_X=0$.
For $|X|=i>k$, suppose that $\mu_Y$ has been computed
for every proper subset $Y$ of $X$, and recall that
$\alpha=k/(i-1)$. For every $a\in X$,
$W\in\supp(\mu_{X-a})$, and $T\subseteq B_a(X)$,
we introduce a variable $p(a,W,T)$ representing the
probability of assigning $T$ to $a$ when $a$ arrives last
among $X$ and the previously assigned elements span $W$.

When $p(a,W,\cdot)$ is a probability distribution for each
$a$ and $W$, define the random subspace $Z_p$ as follows:
choose $a\in X$ uniformly at random, sample
$W\sim\mu_{X-a}$, choose $T\subseteq B_a(X)$ according to
$p(a,W,\cdot)$, and set $Z_p=W+\spn{T}$.
To impose \eqref{eq:xos-invariant} on the distribution of
$Z_p$, define, as in \eqref{eq:violation},
\begin{align*}
g_p(L)
&=i\,\E[\dim(Z_p\cap L)]-(i-k)\dim(L)\\
&=-(i-k)\dim(L)
+\sum_{a\in X}\sum_{W\in\supp(\mu_{X-a})}\mu_{X-a}(W)
\sum_{T\subseteq B_a(X)}p(a,W,T)
\dim((W+\spn{T})\cap L).
\end{align*}
The expanded expression for $g_p(L)$ is affine in $p$.
We impose the following constraints, where $a$ ranges
over $X$, $W$ over $\supp(\mu_{X-a})$, and $T$ over
subsets of $B_a(X)$:
\begin{subequations}
\label{eq:xos-lp}
\begin{alignat}{2}
\mathrm{LP}_{\mathrm{XOS}}(X):\quad
& p(a,W,T)\geq0
&&\qquad\forall\,a,W,T,
\label{eq:xos-nonnegative}\\[2pt]
& \sum_{T\subseteq B_a(X)}p(a,W,T)=1
&&\qquad\forall\,a,W,
\label{eq:xos-normalization}\\[2pt]
& p(a,W,T)=0
&&\qquad\text{if }\dim(W+\spn{T})\neq\dim(W)+|T|,
\label{eq:xos-infeasible}\\[2pt]
& \sum_W\mu_{X-a}(W)
\sum_{\substack{T\subseteq B_a(X)\\e\in T}}p(a,W,T)=\alpha
&&\qquad\forall\,a\in X,\ e\in B_a(X),
\label{eq:xos-marginals}\\[2pt]
& g_p(L)\leq0
&&\qquad\forall\,L\leq V_X.
\label{eq:xos-dimension-constraint}
\end{alignat}
\end{subequations}

Constraints~\eqref{eq:xos-nonnegative}
and~\eqref{eq:xos-normalization} make
$p(a,W,\cdot)$ a probability distribution.
By \eqref{eq:xos-infeasible}, every subset $T$ assigned
with positive probability satisfies
$\dim(W+\spn{T})=\dim(W)+|T|$.
If the previously assigned elements are independent and
span $W$, their union with $T$ is independent, and $T$
contains no previously assigned element.
For each fixed $a\in X$, constraint~\eqref{eq:xos-marginals}
imposes
\[
\Prb[e\in T]
=
\sum_W\mu_{X-a}(W)
\sum_{\substack{T\subseteq B_a(X)\\e\in T}}p(a,W,T)
=\alpha
\qquad\text{for every }e\in B_a(X),
\]
where the probability is over $W\sim\mu_{X-a}$ and the
choice of $T$ according to $p(a,W,\cdot)$.
Finally,~\eqref{eq:xos-dimension-constraint} imposes
\[
\E[\dim(Z_p\cap L)]
\leq\left(1-\frac{k}{i}\right)\dim(L)
\qquad\text{for every }L\leq V_X,
\]
where the expectation also includes the uniform choice
of $a\in X$. Thus the distribution of $Z_p$ satisfies
\eqref{eq:xos-invariant}.

Since the representation is over a finite field, $V_X$
has finitely many subspaces and
$\mathrm{LP}_{\mathrm{XOS}}(X)$ is a finite linear program.
Over an arbitrary field, there are still only finitely
many distinct dimension constraints, as explained in
Remark~\ref{rem:arbitrary-fields}. Each is determined by
$\dim(L)$ and the dimensions of its intersections with
the finitely many spans of subsets of
$\bigcup_{a\in X}E_a$.

Let $P_X$ be the polytope defined by
\eqref{eq:xos-nonnegative}--\eqref{eq:xos-marginals}.
To solve $\mathrm{LP}_{\mathrm{XOS}}(X)$, we must find
$p\in P_X$ such that $g_p(L)\leq0$ for every $L\leq V_X$.

\begin{lemma}
\label{lem:xos-feasibility}
If $\mu_{X-a}$ satisfies \eqref{eq:xos-invariant} for every
$a\in X$, then $\mathrm{LP}_{\mathrm{XOS}}(X)$ is feasible.
\end{lemma}

Similar to Lemma~\ref{lem:lp-feasible}, we prove Lemma~\ref{lem:xos-feasibility} in
Section~\ref{subsec:xos-feasibility} via an uncrossing argument.
Suppose that each $\mu_{X-a}$ satisfies
\eqref{eq:xos-invariant} and has rational entries.
The linear program is feasible by
Lemma~\ref{lem:xos-feasibility}, and its coefficients
are rational, so it has a rational feasible solution.
We compute such a solution from $X$
and the data revealed by its agents, and denote it by $p_X$.

We use $p_X$ to complete the recursive definition of $\mu_X$.
Recall that $Z_p=W+\spn{T}$, where $a$ is uniform in $X$,
$W\sim\mu_{X-a}$, and $T$ is chosen according to
$p(a,W,\cdot)$. Taking $p=p_X$, we set, for every $U\leq V_X$,
\begin{equation}
\begin{aligned}
\mu_X(U)
&=\Prb[Z_{p_X}=U]\\
&=\frac{1}{i}\sum_{a\in X}
\sum_{W\in\supp(\mu_{X-a})}\mu_{X-a}(W)
\sum_{T\subseteq B_a(X)}p_X(a,W,T)
\1{U=W+\spn{T}}.
\end{aligned}
\label{eq:xos-solution-distribution}
\end{equation}
Constraint~\eqref{eq:xos-dimension-constraint} ensures that
$\mu_X$ satisfies \eqref{eq:xos-invariant}.
Since $p_X$ and the distributions $\mu_{X-a}$ have rational
entries, \eqref{eq:xos-solution-distribution} gives rational
entries for $\mu_X$. The base distributions
$\mu_X=\delta_{\{0\}}$ for $|X|\leq k$ have rational entries
and satisfy \eqref{eq:xos-invariant} when $|X|=k$,
so the recursion applies at every size.

\paragraph{XOS secretary algorithm. }
We store $(\mu_Y,p_Y)$ for subsets $Y$ of the observed
agents, together with $\OPT(Y)$ when $|Y|>k$.
When a new agent $a$ arrives,
Algorithm~\ref{alg:xos-update} computes the entries for
sets $Y$ containing $a$, in increasing order of $|Y|$.
For each $b\in Y$, the distribution $\mu_{Y-b}$ has then
either been computed earlier in this update or stored
before $a$ arrived. Hence it is available when we solve
$\mathrm{LP}_{\mathrm{XOS}}(Y)$.
The deterministic choices of $\OPT(Y)$ and $p_Y$ depend
only on $Y$ and its agents' data, so the stored entries
do not depend on the arrival order.

Algorithm~\ref{alg:xos-online} calls
Algorithm~\ref{alg:xos-update} before assigning a bundle
to the current agent $a$. During the sample, it assigns
the empty bundle. Otherwise, it samples
$T_a\subseteq B_a(X)$ according to
$p_X(a,\spn{\overline{\ALG}},\cdot)$.
When $n\geq3$, we run $\mathcal A_k$ with
$k=\lfloor n/e\rfloor\geq1$.
When $n\in\{1,2\}$, we instead assign a maximum-value
independent subset of $E_a$ to the first arriving agent $a$
and empty bundles to the remaining agents.
In Algorithms~\ref{alg:xos-online} and~\ref{alg:xos-update},
we obtain the vectors of $E_a$ from the given representation
or from agent $a$ upon arrival, according to the model.

\begin{algorithm}[ht]
\caption{XOS Matroid Secretary ($\mathcal A_k$)}
\label{alg:xos-online}
\KwIn{Number of agents $n$ and sample size $1\leq k<n$}
$X\gets\emptyset$; $\ALG\gets\emptyset$\;
\ForEach{arriving agent $a$}{
  Observe $E_a$, valuation $f_a$, and the vectors of $E_a$\;
  $X\gets X+a$\;
  \textsc{UpdateAllocationDistributions}$(X,a)$\;
  $T_a\gets\emptyset$\;
  \If{$|X|>k$}{
    Sample $T_a\subseteq B_a(X)$ according to
    $p_X(a,\spn{\overline{\ALG}},\cdot)$\;
  }
  Assign $T_a$ irrevocably to $a$\;
  $\ALG\gets\ALG\cup\{(a,T_a)\}$\;
}
\KwRet{$\ALG$}
\end{algorithm}

\begin{algorithm}[ht]
\caption{\textsc{UpdateAllocationDistributions}$(X,a)$}
\label{alg:xos-update}
\KwIn{The observed agent set $X$, the newly arrived agent $a$,
together with the revealed eligibility sets, valuations,
and vectors of the agents in $X$}
\ForEach{$Y\subseteq X$ with $a\in Y$, in increasing order of $|Y|$}{
  \eIf{$|Y|\leq k$}{
    Set $\mu_Y=\delta_{\{0\}}$ and $p_Y=0$\;
  }{
    Compute $\OPT(Y)=\{(b,B_b(Y)):b\in Y\}$ using
    the deterministic tie-breaking rule for optimal allocations\;
    Enumerate the subspaces of $V_Y$\;
    Solve $\mathrm{LP}_{\mathrm{XOS}}(Y)$ using
    $(\mu_{Y-b})_{b\in Y}$, and compute a rational feasible
    solution $p_Y$ deterministically from $Y$ and its agents' data\;
    Compute $\mu_Y$ using \eqref{eq:xos-solution-distribution}
    for the set $Y$\;
  }
  Store $(\mu_Y,p_Y)$ and, when $|Y|>k$, $\OPT(Y)$\;
}
\end{algorithm}

For single-element arrivals, let each agent $a_e$ have
$E_{a_e}=\{e\}$ and $f_{a_e}(\{e\})=w(e)$, and identify
each observed agent set with the corresponding element set.
Choose $\OPT(X)$ from the same greedy optimum as in
Algorithm~\ref{alg:online}. Eliminating the empty-set
variables from $\mathrm{LP}_{\mathrm{XOS}}(X)$ then gives
$\mathrm{LP}(X)$, with
$p_X(e,W)=p_X(a_e,W,\{e\})$ for
$e\in\overline{\OPT(X)}$ and $p_X(e,W)=0$ otherwise.
If we choose the feasible solutions to correspond under
this identification, Algorithms~\ref{alg:xos-online}
and~\ref{alg:xos-update} recover
Algorithms~\ref{alg:online} and~\ref{alg:update}.
In particular, after each arrival, $\overline{\ALG}$ has
the same distribution as the accepted set in
Algorithm~\ref{alg:online}.

\begin{lemma}
\label{lem:xos-execution}
The algorithm $\mathcal A_k$ is well-defined and always
produces a feasible allocation.
Conditional on $X$ being the set of the first $|X|$
arriving agents, the span $\spn{\overline{\ALG}}$ after
processing $X$ has distribution $\mu_X$, which satisfies
\eqref{eq:xos-invariant} when $|X|\geq k$.
Moreover, for every $i>k$, every set $X$ of $i$ agents,
every $a\in X$, and every $e\in B_a(X)$,
\begin{equation}
\Prb[e\in T_a\mid X_i=X,\ a_i=a]
=\frac{k}{i-1},
\label{eq:xos-online-marginal}
\end{equation}
where $X_i$ is the set of the first $i$ agents and $a_i$
is the agent arriving in position $i$.
\end{lemma}

\begin{proof}
We use the same induction on $|X|$ as in
Lemma~\ref{lem:algorithm-properties}.
If $|X|\leq k$, the algorithm assigns empty bundles,
so the span is $\{0\}$ and has distribution
$\mu_X=\delta_{\{0\}}$.
For $|X|=k$, this distribution satisfies
\eqref{eq:xos-invariant}.
Suppose that $|X|=i>k$.
By induction, each $\mu_{X-a}$ satisfies
\eqref{eq:xos-invariant}.
Lemma~\ref{lem:xos-feasibility} therefore ensures that
$\mathrm{LP}_{\mathrm{XOS}}(X)$ is feasible.
Conditional on $X_i=X$ and $a_i=a$, the agents in $X-a$
arrive in uniformly random order. Thus the span $W$ of
the elements assigned before $a$ arrives has distribution
$\mu_{X-a}$, and the algorithm chooses $T_a$ according to
$p_X(a,W,\cdot)$.

Conditional on $X_i=X$, the last agent $a_i$ is uniform
in $X$. Averaging over $a_i$, $W$, and $T_a$ gives
\eqref{eq:xos-solution-distribution}, so the span after
processing $X$ has distribution $\mu_X$.
Constraint~\eqref{eq:xos-dimension-constraint} gives
\eqref{eq:xos-invariant}, and
constraint~\eqref{eq:xos-marginals} gives
\eqref{eq:xos-online-marginal}.

Finally, $T_a\subseteq B_a(X)\subseteq E_a$.
By induction, the previously assigned elements are
independent and span $W$.
Constraint~\eqref{eq:xos-infeasible} ensures that
$\dim(W+\spn{T_a})=\dim(W)+|T_a|$.
Hence $T_a$ contains no previously assigned element,
and the union of all assigned bundles remains independent
in $M$.
\end{proof}

\subsection{Feasibility via Auxiliary Vectors and Uncrossing}
\label{subsec:xos-feasibility}

Fix $X$ with $|X|=i>k$ and distributions
$(\mu_{X-a})_{a\in X}$ satisfying \eqref{eq:xos-invariant}, and
write $\alpha=k/(i-1)$. Applying Lemma~\ref{lem:xos-marginals}
to each $B_a(X)$ with $W\sim\mu_{X-a}$ shows that $P_X$ is
nonempty.

As in the proof of Lemma~\ref{lem:chains}, we first fix
a chain of subspaces and find a point $p\in P_X$ satisfying
its dimension constraints. We may choose a different $p$
for each chain. Once we can do this for every chain,
the same uncrossing argument gives a single $p\in P_X$
satisfying the constraints for all subspaces.
In Section~\ref{sec:feasibility}, we constructed the
acceptance probabilities by ordering the predecessor
subspaces $W$. Here, we add auxiliary vectors to each
bundle $B_a(X)$ and apply Lemma~\ref{lem:xos-marginals}
to the larger independent set.

\begin{lemma}
\label{lem:xos-chain}
Let $\alpha\in[0,1]$, and let $W$ be a random subspace of $V$
such that, for every $L\leq V$,
\[
\E[\dim(W\cap L)]
\leq (1-\alpha)\dim(L).
\]
Let $B$ be an independent set of vectors in $V$, and put
$H=\spn{B}$. For every chain $L_1<\cdots<L_s\leq V$,
there is a randomized rule that, after observing $W$,
selects $T\subseteq B$ satisfying
$\dim(W+\spn{T})=\dim(W)+|T|$ and
$\Prb[b\in T]=\alpha$ for every $b\in B$.
Moreover, for every $j\in\{1,\dots,s\}$,
\[
\E[\dim((W+\spn{T})\cap L_j)]
\leq (1-\alpha)\dim(L_j)+\alpha\dim(H\cap L_j).
\]
\end{lemma}

\begin{proof}
To bound $\dim((W+\spn{T})\cap L_j)$, we add auxiliary
vectors in each $L_j$ and select them together with $T$.
We require the full selected set to be independent modulo $W$.
The selected auxiliary vectors in $L_j$ then span a subspace
with trivial intersection with $W+\spn{T}$.
Thus $\dim((W+\spn{T})\cap L_j)$ is at most $\dim(L_j)$
minus the number of selected auxiliary vectors in $L_j$.

Extend $B$ to a basis of $H+L_1$ by adding vectors from
$L_1$. Continue along the chain, extending the basis of
$H+L_{j-1}$ to a basis of $H+L_j$ by adding vectors from
$L_j$. This gives nested sets
$D_1\subseteq\cdots\subseteq D_s$ such that
$D_j\subseteq L_j$ and $B\cup D_j$ is a basis of $H+L_j$.
In particular,
\[
|D_j|=\dim(L_j)-\dim(H\cap L_j).
\]

The auxiliary vectors in $D_s$ need not correspond to
elements of $E$. Lemma~\ref{lem:xos-marginals} applies to
any independent set of vectors in $V$, so we can apply
it to $B\cup D_s$.
Let $R\subseteq B\cup D_s$ be the selected subset.
We retain only $T=R\cap B$ and put $W'=W+\spn{T}$.

Every $b\in B$ satisfies $\Prb[b\in T]=\alpha$.
Lemma~\ref{lem:xos-marginals} also gives
$\dim(W+\spn{R})=\dim(W)+|R|$.
Thus $\dim(W+\spn{T})=\dim(W)+|T|$, and the selected
auxiliary vectors in $R\cap D_j$ are linearly independent
and satisfy $W'\cap\spn{R\cap D_j}=\{0\}$.
Since $R\cap D_j\subseteq L_j$, the subspaces
$W'\cap L_j$ and $\spn{R\cap D_j}$ form a direct sum
inside $L_j$, so
\[
\dim(W'\cap L_j)\leq\dim(L_j)-|R\cap D_j|.
\]
Each vector of $D_j$ is selected with probability
$\alpha$, so $\E[|R\cap D_j|]=\alpha|D_j|$.
Taking expectations gives
\begin{align*}
\E[\dim(W'\cap L_j)]
&\leq \dim(L_j)-\alpha|D_j|=(1-\alpha)\dim(L_j)+\alpha\dim(H\cap L_j).\qedhere
\end{align*}
\end{proof}

\begin{proof}[Proof of Lemma~\ref{lem:xos-feasibility}]
Fix a chain $L_1<\cdots<L_s\leq V_X$.
For each $a\in X$, apply Lemma~\ref{lem:xos-chain} to
$B_a(X)$ with $W\sim\mu_{X-a}$ and $\alpha=k/(i-1)$.
Let $p(a,W,T)$ be the conditional probability that
the selection rule assigns $T$, given $W$.
The resulting probabilities satisfy
\eqref{eq:xos-nonnegative}--\eqref{eq:xos-marginals},
so $p\in P_X$.

Put $H_a=\spn{B_a(X)}$. Since $\OPT(X)$ is feasible, its
bundles are pairwise disjoint and $\overline{\OPT(X)}$ is
independent in $M$. Thus the subspaces $H_a$ form a direct sum.
For every $L\leq V_X$, the
subspaces $H_a\cap L$ also form a direct sum contained in $L$,
so
\begin{equation}
\sum_{a\in X}\dim(H_a\cap L)\leq\dim(L).
\label{eq:xos-bundle-dimensions}
\end{equation}
When $B_a(X)=\{e\}$, we have
$\dim(H_a\cap L)=\1{e\in L}$.
Thus \eqref{eq:xos-bundle-dimensions} extends the bound
on the number of optimal elements lying in $L$ used
at the end of Section~\ref{sec:feasibility}.
Averaging the bound of Lemma~\ref{lem:xos-chain} over
the uniformly chosen agent $a\in X$ and using
\eqref{eq:xos-bundle-dimensions}, we obtain, for every
$L$ in the chain,
\begin{align*}
\E[\dim(Z_p\cap L)]
&\leq (1-\alpha)\dim(L)
+\frac{\alpha}{i}\sum_{a\in X}\dim(H_a\cap L)\\
&\leq \left(1-\alpha+\frac{\alpha}{i}\right)\dim(L)
=\left(1-\frac{k}{i}\right)\dim(L),
\end{align*}
so $g_p(L)\leq 0$ throughout the chain.

We now apply the proof of Lemma~\ref{lem:chains}, with
$P_X$ in place of $P$ and $V_X$ in place of $\spn{X}$.
Since the field is finite, $V_X$ has finitely many subspaces.
The polytope $P_X$ is nonempty, compact, and convex,
and $g_p(L)$ is affine in $p$.
For each fixed $p$, the function $g_p$ is supermodular:
$L\mapsto\dim(Z_p\cap L)$ is supermodular by the calculation
in Section~\ref{sec:feasibility}, taking expectations preserves
this property, and $\dim(L)$ is modular.
We have constructed a point of $P_X$ satisfying the dimension
constraints for each fixed chain. The proof of
Lemma~\ref{lem:chains} therefore gives $p\in P_X$
such that $g_p(L)\leq0$ for all $L\leq V_X$.
\end{proof}

\subsection{The \texorpdfstring{$1/e$}{1/e} Guarantees}
\label{subsec:xos-guarantees}

Recall that $X_i$ is the set of the first $i$ agents
and $a_i$ is the agent arriving in position $i$.
For XOS valuations, we first bound the expected value
assigned at position $i$ using $\OPT(X_i)$.
We then compare $F(\OPT(X_i))$ with the value obtained
by restricting the fixed allocation $\OPT(N)$ to $X_i$.

\begin{proof}[Proof of Theorem~\ref{thm:xos}]
Suppose first that a finite-field representation is supplied
and $n\geq3$. Run $\mathcal A_k$ with $k=\lfloor n/e\rfloor$.
Fix $i>k$ and a set $X$ of $i$ agents. Let
$R_i=f_{a_i}(T_{a_i})$ be the value assigned at position $i$.
Condition on $X_i=X$ and $a_i=a$. The bundle $B_a(X)$
is now fixed, and the algorithm assigns each of its
elements with probability $k/(i-1)$ by
\eqref{eq:xos-online-marginal}. Since $f_a$ is XOS,
we apply \eqref{eq:xos-bundle-value} under this
conditioning to obtain
$\E[f_a(T_a)\mid X_i=X,a_i=a]
\geq \frac{k}{i-1}f_a(B_a(X))$.

Conditional on $X_i=X$, the last agent $a_i$ is uniform
in $X$. Averaging over $a$ gives
\begin{align}
\E[R_i\mid X_i=X]
&=\frac{1}{i}\sum_{a\in X}
\E[f_a(T_a)\mid X_i=X,\ a_i=a]\nonumber\\
&\geq\frac{k}{i(i-1)}
\sum_{a\in X}f_a(B_a(X))
=\frac{k}{i-1}\frac{F(\OPT(X))}{i}.
\label{eq:xos-prefix-value}
\end{align}
Write $\OPT(N)=\{(a,B_a^*):a\in N\}$ for the fixed optimal
allocation for all agents. Restricting $\OPT(N)$ to the agents
of $X_i$ yields a feasible allocation for $X_i$, and each
agent $a\in N$ belongs to $X_i$
with probability $i/n$, so
\begin{equation}
\E[F(\OPT(X_i))]
\geq
\E\left[\sum_{a\in X_i}f_a(B_a^*)\right]
=\frac{i}{n}F(\OPT(N)).
\label{eq:xos-prefix-optimum}
\end{equation}
Taking expectations in \eqref{eq:xos-prefix-value},
using \eqref{eq:xos-prefix-optimum}, and summing over
$i\in\{k+1,\dots,n\}$ gives
\[
\E[F(\ALG)]
\geq
\frac{k}{n}\sum_{j=k}^{n-1}\frac{1}{j}\,F(\OPT(N))
\geq \frac{1}{e}F(\OPT(N)).
\]
The last inequality is the harmonic-sum bound proved after
\eqref{eq:ratio} in Section~\ref{sec:guarantee}
for $k=\lfloor n/e\rfloor$.

For $n\in\{1,2\}$, assigning a maximum-value independent subset
of $E_a$ to the first arriving agent $a$ and empty bundles to
the remaining agents gives
\[
\E[F(\ALG)]\geq
\frac{1}{n}\sum_{a\in N}f_a(B_a^*)
=\frac{1}{n}F(\OPT(N))
\geq \frac{1}{e}F(\OPT(N)).
\]
In the known-matroid model, if no finite-field representation
is supplied, we compute one before the arrivals by the
enumeration argument in the proof of
Corollary~\ref{cor:known-linear} and run the same algorithm.
The same algorithm also works when a representation over a
known finite field is revealed with the agents, in a common
coordinate system. For every observed set $Y$, we compute
$\OPT(Y)$ and $V_Y$ from the eligibility sets, valuations,
and vectors revealed by the agents in $Y$.
All computations in Algorithm~\ref{alg:xos-update} therefore
use information available at that point.
\end{proof}

For group arrivals, we fix an element of the global
greedy optimum. By the argument in
Section~\ref{sec:guarantee}, this element belongs to
the greedy optimum of every prefix containing its group.
We then bound its probability of acceptance.
\begin{proof}[Proof of Corollary~\ref{cor:xos-groups}]
Let $I^*=\overline{\OPT(N)}=I(N)$ be the global greedy optimum
under the fixed tie-breaking order. For every $e\in I^*$ and
every set of groups $X$ containing its group, we have $e\in I(X)$.
Indeed, $e$ is not spanned by the elements preceding it in the
global greedy order. Restricting those elements to
$\bigcup_{a\in X}G_a$ cannot make them span $e$.

Suppose $n\geq3$ and run $\mathcal A_k$ with
$k=\lfloor n/e\rfloor$ and bundles $B_a(X)=I(X)\cap G_a$.
Fix $e\in I^*$ and let $a$ be its group. Conditional on
$X_i=X$ and $a_i=a$ for $i>k$, we have $e\in B_a(X)$, so
Lemma~\ref{lem:xos-execution} gives
$\Prb[e\in T_a\mid X_i=X,\ a_i=a]=k/(i-1)$.
This probability is the same for every possible $X$,
so $\Prb[e\in T_a\mid a_i=a]=k/(i-1)$.
Averaging over the uniformly distributed arrival
position of group $a$ gives
\[
\Prb[e\in\overline{\ALG}]
=\frac{k}{n}\sum_{j=k}^{n-1}\frac{1}{j}
\geq \frac{1}{e}.
\]
For $n\in\{1,2\}$, accepting the greedy independent set of the
first arriving group and rejecting the remaining groups accepts
each $e\in I^*$ whenever its group arrives first, with
probability $1/n\geq1/e$.

In the online-representation model, for every observed subset
$Y$, both the greedy optimum $I(Y)$ and the ambient subspace
$V_Y=\spn{\bigcup_{a\in Y}G_a}$ used in
Algorithm~\ref{alg:xos-update} depend only on the groups in
$Y$ and their revealed vectors and weights. Thus
Algorithms~\ref{alg:xos-online} and~\ref{alg:xos-update} can be
executed online without prior knowledge of the full matroid,
with the same acceptance probabilities.
\end{proof}

\begin{proof}[Proof of Corollary~\ref{cor:xos-agents}]
In the unit-demand formulation, we chose each $\OPT(X)$
so that $|B_a(X)|\leq1$ for every $a\in X$.
Since $\mathcal A_k$ assigns only subsets of $B_a(X)$,
every agent receives at most one element, and
Theorem~\ref{thm:xos} gives the $1/e$ expected-value guarantee
(with the same singleton restriction on the first agent's
bundle when $n\in\{1,2\}$).
\end{proof}

\medskip

\paragraph{Acknowledgments. }
The research received support from the Lend\"ulet Programme of the Hungarian Academy of Sciences (LP2021-1/2021), from the Ministry of Innovation and Technology of Hungary from the National Research, Development and Innovation Fund (ADVANCED 150556, 153096), from the Dynasnet European Research Council Synergy project (ERC-2018-SYG 810115) and from ANID-Chile through research grants CMM FB210005 and FONDECYT 1231669. Dughmi and Livanos were also supported by the Air Force Office of Scientific Research under award number FA9550-24-1-0261.

Part of this work was done while Dughmi was on sabbatical as the Carter and Tania Neild visiting professor at Northwestern University, as well as a visiting professor in the Data Science Institute at the University of Chicago and Livanos was at Universidad de Chile and at the Center for Mathematical Modeling CNRS-IRL 2807, Chile.

\section*{AI Methodology.}

We had been working for some time on an approach to the Matroid Secretary Problem that assigns correlated acceptance probabilities equal to $k/(i-1)$ in order to achieve a $1/e$ competitive ratio, while maintaining a subspace dimension invariant. Motivated by Buchbinder, Jain, and Singh~\cite{BuchbinderJS14}, we formulated this step by a linear program, focusing primarily on binary matroids. While trying to prove the feasibility of this LP, we discussed the problem with ChatGPT-6 Astra, which identified the supermodularity of $g_p$ and proposed the uncrossing argument used to establish feasibility. The resulting proof, which forms the basis of Section \ref{sec:feasibility}, was initially drafted with assistance from ChatGPT-6 Astra. Building on this feasibility argument, we developed the extensions in Sections~\ref{sec:extensions} and~\ref{sec:xos}. All arguments and references were checked by the authors, who are responsible for their verification and final presentation.

\bibliographystyle{abbrvurl}

\bibliography{refs}

@inproceedings{BercziEtAl,
  author       = {Krist{\'{o}}f B{\'{e}}rczi and Bogl{\'{a}}rka Geh{\'{e}}r and Andr{\'{a}}s Imolay and L{\'{a}}szl{\'{o}} Lov{\'{a}}sz and Carles Padr{\'{o}} and Tam{\'{a}}s Schwarcz},
  title        = {Interaction between skew-representability, tensor products, extension properties, and rank inequalities},
  booktitle    = {SODA},
  pages        = {328--354},
  year         = {2026}
}

@book{Schrijver2003,
  author       = {Alexander Schrijver},
  title        = {Combinatorial Optimization: Polyhedra and Efficiency},
  series       = {Algorithms and Combinatorics},
  volume       = {24},
  publisher    = {Springer},
  year         = {2003},
  isbn         = {978-3-540-44389-6}
}

@article{Rado1957,
  author       = {Rado, R.},
  title        = {Note on independence functions},
  journal      = {Proceedings of the London Mathematical Society},
  volume       = {7},
  pages        = {300--320},
  year         = {1957}
}

@article{lindley,
  author       = {D. V. Lindley},
  title        = {Dynamic Programming and Decision Theory},
  journal      = {Journal of the Royal Statistical Society. Series C (Applied Statistics)},
  volume       = {10},
  number       = {1},
  pages        = {39--51},
  year         = {1961}
}

@article{dynkin,
  author       = {Dynkin, E. B.},
  title        = {The optimum choice of the instant for stopping a {M}arkov process},
  journal      = {Soviet Mathematics Doklady},
  volume       = {4},
  pages        = {627--629},
  year         = {1963}
}

@inproceedings{klein-secr,
  author       = {Robert Kleinberg},
  title        = {A Multiple-Choice Secretary Algorithm with Applications to Online Auctions},
  booktitle    = {SODA},
  pages        = {630--631},
  year         = {2005}
}

@inproceedings{sqrt-log-r-matroid-sec,
  author       = {Sourav Chakraborty and Oded Lachish},
  title        = {Improved competitive ratio for the matroid secretary problem},
  booktitle    = {SODA},
  pages        = {1702--1712},
  year         = {2012}
}

@inproceedings{loglogrank-matroid-sec1,
  author       = {Oded Lachish},
  title        = {{$O(\log \log \mathrm{Rank})$} Competitive Ratio for the Matroid Secretary Problem},
  booktitle    = {FOCS},
  pages        = {326--335},
  year         = {2014}
}

@article{loglogrank-matroid-sec2,
  author       = {Moran Feldman and Ola Svensson and Rico Zenklusen},
  title        = {A Simple {$O(\log \log(\mathrm{rank}))$}-Competitive Algorithm for the Matroid Secretary Problem},
  journal      = {Mathematics of Operations Research},
  volume       = {43},
  number       = {2},
  pages        = {638--650},
  year         = {2018}
}

@article{bateni-secretary,
  author       = {MohammadHossein Bateni and MohammadTaghi Hajiaghayi and Morteza Zadimoghaddam},
  title        = {Submodular secretary problem and extensions},
  journal      = {ACM Transactions on Algorithms},
  volume       = {9},
  number       = {4},
  pages        = {32:1--32:23},
  year         = {2013}
}

@article{forbidden-paper,
  author       = {Jos{\'{e}} A. Soto and Abner Turkieltaub and Victor Verdugo},
  title        = {Strong algorithms for the ordinal matroid secretary problem},
  journal      = {Mathematics of Operations Research},
  volume       = {46},
  number       = {2},
  pages        = {642--673},
  year         = {2021}
}

@article{trans-secretary,
  author       = {Nedialko B. Dimitrov and C. Greg Plaxton},
  title        = {Competitive Weighted Matching in Transversal Matroids},
  journal      = {Algorithmica},
  volume       = {62},
  number       = {1},
  pages        = {333--348},
  year         = {2012}
}

@inproceedings{KesselheimRTV13,
  author       = {Thomas Kesselheim and Klaus Radke and Andreas T{\"{o}}nnis and Berthold V{\"{o}}cking},
  title        = {An Optimal Online Algorithm for Weighted Bipartite Matching and Extensions to Combinatorial Auctions},
  booktitle    = {ESA},
  pages        = {589--600},
  year         = {2013}
}

@inproceedings{babaioff-secretary-2,
  author       = {Moshe Babaioff and Michael Dinitz and Anupam Gupta and Nicole Immorlica and Kunal Talwar},
  title        = {Secretary Problems: Weights and Discounts},
  booktitle    = {SODA},
  pages        = {1245--1254},
  year         = {2009}
}

@inproceedings{korula-pal-graphic,
  author       = {Nitish Korula and Martin P{\'{a}}l},
  title        = {Algorithms for secretary problems on graphs and hypergraphs},
  booktitle    = {ICALP},
  pages        = {508--520},
  year         = {2009}
}

@inproceedings{ImW11,
  author       = {Sungjin Im and Yajun Wang},
  title        = {Secretary Problems: Laminar Matroid and Interval Scheduling},
  booktitle    = {SODA},
  pages        = {1265--1274},
  year         = {2011}
}

@inproceedings{free-order-secretary,
  author       = {Patrick Jaillet and Jos{\'{e}} A. Soto and Rico Zenklusen},
  title        = {Advances on Matroid Secretary Problems: Free Order Model and Laminar Case},
  booktitle    = {IPCO},
  pages        = {254--265},
  year         = {2013}
}

@inproceedings{ma-tang-secretary,
  author       = {Tengyu Ma and Bo Tang and Yajun Wang},
  title        = {The Simulated Greedy Algorithm for Several Submodular Matroid Secretary Problems},
  booktitle    = {STACS},
  pages        = {478--489},
  year         = {2013}
}

@article{dinitz-secretary,
  author       = {Michael Dinitz and Guy Kortsarz},
  title        = {Matroid Secretary for Regular and Decomposable Matroids},
  journal      = {SIAM Journal on Computing},
  volume       = {43},
  number       = {5},
  pages        = {1807--1830},
  year         = {2014}
}

@article{GharanV13,
  author       = {Shayan {Oveis Gharan} and Jan Vondr{\'{a}}k},
  title        = {On Variants of the Matroid Secretary Problem},
  journal      = {Algorithmica},
  volume       = {67},
  number       = {4},
  pages        = {472--497},
  year         = {2013}
}

@inproceedings{matroid-embeddings,
  author       = {Andr{\'{e}}s Cristi and Paul D{\"{u}}tting and Robert Kleinberg and Renato {Paes Leme} and Neel Patel},
  title        = {Online Matroid Embeddings},
  booktitle    = {ICALP},
  pages        = {69:1--69:24},
  year         = {2026}
}

@inproceedings{dughmi1,
  author       = {Shaddin Dughmi},
  title        = {The Outer Limits of Contention Resolution on Matroids and Connections to the Secretary Problem},
  booktitle    = {ICALP},
  pages        = {42:1--42:18},
  year         = {2020}
}

@inproceedings{dughmi2,
  author       = {Shaddin Dughmi},
  title        = {Matroid Secretary Is Equivalent to Contention Resolution},
  booktitle    = {ITCS},
  pages        = {58:1--58:23},
  year         = {2022}
}

@inproceedings{BahraniBSW21,
  author       = {Maryam Bahrani and Hedyeh Beyhaghi and Sahil Singla and S. Matthew Weinberg},
  title        = {Formal Barriers to Simple Algorithms for the Matroid Secretary Problem},
  booktitle    = {WINE},
  pages        = {280--298},
  year         = {2021}
}

@inproceedings{matroid-partition,
  author       = {Dorna Abdolazimi and Anna R. Karlin and Nathan Klein and Shayan {Oveis Gharan}},
  title        = {Matroid Partition Property and the Secretary Problem},
  booktitle    = {ITCS},
  pages        = {2:1--2:9},
  year         = {2023}
}

@article{soto-secretary,
  author       = {Jos{\'{e}} A. Soto},
  title        = {Matroid Secretary Problem in the Random-Assignment Model},
  journal      = {SIAM Journal on Computing},
  volume       = {42},
  number       = {1},
  pages        = {178--211},
  year         = {2013}
}

@inproceedings{correa-googol,
  author       = {Jos{\'{e}} R. Correa and Andr{\'{e}}s Cristi and Boris Epstein and Jos{\'{e}} A. Soto},
  title        = {The Two-Sided Game of Googol and Sample-Based Prophet Inequalities},
  booktitle    = {SODA},
  pages        = {2066--2081},
  year         = {2020}
}

@inproceedings{BIK2007,
  author       = {Moshe Babaioff and Nicole Immorlica and Robert Kleinberg},
  title        = {Matroids, secretary problems, and online mechanisms},
  booktitle    = {SODA},
  pages        = {434--443},
  year         = {2007}
}

@article{mat-sec,
  author       = {Moshe Babaioff and Nicole Immorlica and David Kempe and Robert Kleinberg},
  title        = {Matroid Secretary Problems},
  journal      = {Journal of the ACM},
  volume       = {65},
  number       = {6},
  pages        = {35:1--35:26},
  year         = {2018}
}

@inproceedings{labeling-schemes,
  author       = {Krist{\'{o}}f B{\'{e}}rczi and Vasilis Livanos and Jos{\'{e}} A. Soto and Victor Verdugo},
  title        = {Matroid Secretary via Labeling Schemes},
  booktitle    = {IPCO},
  pages        = {128--141},
  year         = {2025}
}

@inproceedings{banihashem2025beating,
  author       = {Kiarash Banihashem and MohammadTaghi Hajiaghayi and Dariusz R. Kowalski and Piotr Krysta and Danny Mittal and Jan Olkowski},
  title        = {Beating Competitive Ratio 4 for Graphic Matroid Secretary},
  booktitle    = {ESA},
  pages        = {52:1--52:16},
  year         = {2025}
}

@inproceedings{zahra-laminar-secretary,
  author       = {Zhiyi Huang and Zahra Parsaeian and Zixuan Zhu},
  title        = {Laminar Matroid Secretary: Greedy Strikes Back},
  booktitle    = {ESA},
  pages        = {73:1--73:8},
  year         = {2024}
}

@article{SantiagoSZ23,
  author       = {Richard Santiago and Ivan Sergeev and Rico Zenklusen},
  title        = {Constant-competitiveness for random assignment Matroid secretary without knowing the Matroid},
  journal      = {Mathematical Programming},
  volume       = {210},
  number       = {1},
  pages        = {815--846},
  year         = {2025}
}

@article{FeldmanZenklusen2018,
  author       = {Moran Feldman and Rico Zenklusen},
  title        = {The Submodular Secretary Problem Goes Linear},
  journal      = {SIAM Journal on Computing},
  volume       = {47},
  number       = {2},
  pages        = {330--366},
  year         = {2018}
}

@article{FeldmanSZ22,
  author       = {Moran Feldman and Ola Svensson and Rico Zenklusen},
  title        = {A Framework for the Secretary Problem on the Intersection of Matroids},
  journal      = {SIAM Journal on Computing},
  volume       = {51},
  number       = {3},
  pages        = {766--819},
  year         = {2022}
}

@article{BuchbinderJS14,
  author       = {Niv Buchbinder and Kamal Jain and Mohit Singh},
  title        = {Secretary Problems via Linear Programming},
  journal      = {Mathematics of Operations Research},
  volume       = {39},
  number       = {1},
  pages        = {190--206},
  year         = {2014}
}

@inproceedings{ChanCJ15,
  author       = {T.-H. Hubert Chan and Fei Chen and Shaofeng H.-C. Jiang},
  title        = {Revealing Optimal Thresholds for Generalized Secretary Problem via Continuous {LP}: Impacts on Online {K}-Item Auction and Bipartite {K}-Matching with Random Arrival Order},
  booktitle    = {SODA},
  pages        = {1169--1188},
  year         = {2015}
}

@article{DuttingLLV24,
  author       = {Paul D{\"{u}}tting and Silvio Lattanzi and Renato {Paes Leme} and Sergei Vassilvitskii},
  title        = {Secretaries with Advice},
  journal      = {Mathematics of Operations Research},
  volume       = {49},
  number       = {2},
  pages        = {856--879},
  year         = {2024}
}

@article{CorreaCES24,
  author       = {Jos{\'{e}} Correa and Andr{\'{e}}s Cristi and Boris Epstein and Jos{\'{e}} A. Soto},
  title        = {Sample-Driven Optimal Stopping: From the Secretary Problem to the i.i.d. Prophet Inequality},
  journal      = {Mathematics of Operations Research},
  volume       = {49},
  number       = {1},
  pages        = {441--475},
  year         = {2024}
}

@article{marinkovic-soto-verdugo,
  author       = {Javier Marinkovic and Jos{\'{e}} A. Soto and Victor Verdugo},
  title        = {Online combinatorial assignment in independence systems},
  journal      = {Mathematical Programming},
  volume       = {216},
  number       = {1},
  pages        = {49--86},
  year         = {2026}
}

@article{LehmannLN06,
  author       = {Benny Lehmann and Daniel Lehmann and Noam Nisan},
  title        = {Combinatorial auctions with decreasing marginal utilities},
  journal      = {Games and Economic Behavior},
  volume       = {55},
  number       = {2},
  pages        = {270--296},
  year         = {2006}
}

@article{Feige09,
  author       = {Uriel Feige},
  title        = {On Maximizing Welfare When Utility Functions Are Subadditive},
  journal      = {SIAM Journal on Computing},
  volume       = {39},
  number       = {1},
  pages        = {122--142},
  year         = {2009}
}

@inproceedings{RubinsteinS17,
  author       = {Aviad Rubinstein and Sahil Singla},
  title        = {Combinatorial Prophet Inequalities},
  booktitle    = {SODA},
  pages        = {1671--1687},
  year         = {2017}
}

@inproceedings{EhsaniHKS18,
  author       = {Soheil Ehsani and MohammadTaghi Hajiaghayi and Thomas Kesselheim and Sahil Singla},
  title        = {Prophet Secretary for Combinatorial Auctions and Matroids},
  booktitle    = {SODA},
  pages        = {700--714},
  year         = {2018}
}

@article{KesselheimRTV18,
  author       = {Thomas Kesselheim and Klaus Radke and Andreas T{\"{o}}nnis and Berthold V{\"{o}}cking},
  title        = {Primal Beats Dual on Online Packing {LPs} in the Random-Order Model},
  journal      = {SIAM Journal on Computing},
  volume       = {47},
  number       = {5},
  pages        = {1939--1964},
  year         = {2018}
}

@inproceedings{DuttingKLRS24,
  author       = {Paul D{\"{u}}tting and Thomas Kesselheim and Brendan Lucier and Rebecca Reiffenh{\"{a}}user and Sahil Singla},
  title        = {Online Combinatorial Allocations and Auctions with Few Samples},
  booktitle    = {FOCS},
  pages        = {1231--1250},
  year         = {2024}
}

@misc{DuttingEtAl2026,
  author        = {Paul D{\"{u}}tting and Renato {Paes Leme} and Martin P{\'{a}}l and Neel Patel},
  title         = {Graphic {M}atroid {S}ecretary without the Graph},
  year          = {2026},
  eprint        = {2608.11413},
  archivePrefix = {arXiv},
  primaryClass  = {cs.DS}
}

@misc{Singla2026,
  author        = {Sahil Singla},
  title         = {The {M}atroid {S}ecretary {C}onjecture is True},
  year          = {2026},
  eprint        = {2609.14555},
  archivePrefix = {arXiv},
  primaryClass  = {cs.DS}
}

@misc{abdi2026strong,
  author        = {Hamed Abdi and Kiarash Banihashem and MohammadTaghi Hajiaghayi and Danny Mittal},
  title         = {On the {S}trong {M}atroid {S}ecretary {C}onjecture and Beyond},
  year          = {2026},
  eprint        = {2609.19118},
  archivePrefix = {arXiv},
  primaryClass  = {cs.DS}
}

@misc{ChanLW2026,
  author        = {Hau Chan and Jianan Lin and Chenhao Wang},
  title         = {A 3.1462-Competitive Algorithm for {M}atroid {S}ecretary},
  year          = {2026},
  eprint        = {2609.17782},
  archivePrefix = {arXiv},
  primaryClass  = {cs.DS}
}

@misc{Huang2026,
  author        = {Zhiyi Huang},
  title         = {Settling the {M}atroid {S}ecretary {P}roblem},
  year          = {2026},
  eprint        = {2609.30421},
  archivePrefix = {arXiv},
  primaryClass  = {cs.DS}
}

\end{document}